\documentclass[11pt,reqno,a4paper]{amsart}
\usepackage{amssymb}
\usepackage{amsmath}
\usepackage[english,activeacute]{babel}
\usepackage{appendix}
\usepackage{xcolor}
\usepackage[active]{srcltx}
\usepackage{verbatim}
 \newtheorem{thm}{Theorem}[section]
 
 \newtheorem{cor}[thm]{Corollary}
 \newtheorem{lem}[thm]{Lemma}
 \newtheorem{prop}[thm]{Proposition}
 \newtheorem{hypothesis}[thm]{Hypothesis}
 \theoremstyle{definition}
 \newtheorem{defn}[thm]{Definition}
 
 \theoremstyle{remark}
 \newtheorem{rem}[thm]{Remark}

 \numberwithin{equation}{section}

\newcommand{\be}{\begin{equation}}
\newcommand{\ee}{\end{equation}}
\newcommand{\benon}{\begin{equation*}}
\newcommand{\eenon}{\end{equation*}}
\newcommand{\ba}{\begin{array}}
\newcommand{\ea}{\end{array}}
\newcommand{\bal}{\begin{align}}
\newcommand{\eal}{\end{align}}
\newcommand{\bea}{\begin{eqnarray}}
\newcommand{\eea}{\end{eqnarray}}
\newcommand{\bee}{\begin{eqnarray*}}
\newcommand{\eee}{\end{eqnarray*}}

\newcommand{\Z}{\mathbb Z}

\newcommand{\R}{\mathbb R}

\newcommand{\E}{\mathbb E}
\newcommand{\N}{\mathbb N}

\newcommand{\norm}[1]{\Vert #1 \Vert}
\newcommand{\abs}[1]{\left| #1 \right|}

\newcommand{\angles}[1]{\langle #1 \rangle}
\numberwithin{equation}{section}

\newcommand{\Schr}{Schr\"odinger }

\numberwithin{equation}{section}

\let\Im\undefined 
\DeclareMathOperator{\Im}{\mathrm{Im}}

\newcommand{\hm}[1]{\leavevmode{\marginpar{\tiny%
$\hbox to 0mm{\hspace*{-0.5mm}$\leftarrow$\hss}%
\vcenter{\vrule depth 0.1mm height 0.1mm width \the\marginparwidth}%
\hbox to 0mm{\hss$\rightarrow$\hspace*{-0.5mm}}$\\\relax\raggedright #1}}}

\newcommand{\tr}{{\mathop{\mathrm{tr} \,}}}

\newcommand{\dist}{\mathop{\mathrm{dist}}}

\newcommand{\supp}{{\mathop{\mathrm{supp\,}}}}

\newcommand{\ceil}[1]{\lceil #1 \rceil}
\newcommand{\floor}[1]{\lfloor #1 \rfloor}

\def\beq{\begin{equation}}
\def\eeq{\end{equation}}
\newcommand{\beas}{\begin{eqnarray*}}
\newcommand{\eeas}{\end{eqnarray*}}

\begin{document}

\title[Eigenvalue statistics and absence of dynamical localization]{Poisson eigenvalue statistics and dynamical delocalization for fractional and long-range Anderson models}

\author[P.\ D.\ Hislop]{Peter D.\ Hislop}
\address{Department of Mathematics,
    University of Kentucky,
    Lexington, Kentucky  40506-0027, USA}
\email{peter.hislop@uky.edu}

\author[R.\ Matos]{Rodrigo Matos}
\address{Department of Mathematics,  
PUC-Rio,
22451-900, Rio de Janeiro, Brasil}
\email{rodrigo@mat.puc-rio.br}

\author[C.\ Rojas-Molina]{Constanza Rojas-Molina}
\address{Laboratoire Analyse, G\'eom\'etrie et Mod\'elisation (AGM),
CY Cergy Paris Universit\'e,
2 av.\ Adolphe Chauvin,
95302 Cergy-Pontoise Cedex,
France}
\email{crojasmo@u-cergy.fr}



\begin{abstract}
We study local eigenvalue statistics and dynamical (de)-localization properties for long-range Anderson models, including the fractional Anderson model. These systems are random Anderson-type perturbations of operators with long-range (non random) hopping terms of the form $\abs{T(n,m)}\sim \|n-m\|^{-(d+2\beta)}$ for $\beta>0$, on the $d$-dimensional lattice. In the presence of a strong enough random potential, these models exhibit dense pure point spectrum with polynomially decaying eigenvectors, almost surely. We show that in this strong disorder regime, and for $\beta>d/2$, the local eigenvalue statistics centered at any $E$ in the deterministic spectrum is a Poisson point process with intensity given by the density of states function $n(E)$. Moreover, we prove that, at strong disorder, there is no dynamical localization for these models,  verifying a conjecture of Disertori et al.  We achieve this by establishing explicit, sharp lower bounds on the Green's functions fractional moments which imply that  large moments of the position operator are infinite. Hence, these models exhibit Poisson eigenvalue statistics and dynamical delocalization, that is, the absence of dynamical localization. In particular, this shows that in dimension $d=1$, the fractional Anderson model, a random perturbation of the fractional Laplacian $(-\Delta)^\alpha$ with $\alpha\in (0,1)$, exhibits Poisson eigenvalue statistics and dynamical delocalization if the disorder is strong and the exponent $\frac{1}{2}<\alpha<1$. This is in stark contrast with what is known for the usual Anderson model that has only nearest-neighbor hopping.

\end{abstract}




\maketitle

\tableofcontents



\section{Introduction and main results}\label{sec:intro1}

In this article, we consider the \emph{long-range Anderson Model} (lrAM) for which the discrete Laplacian is replaced by a long-range, slowly decaying hopping operator. The basic model is given by 
\be\label{op-h-beta} H_{\omega,\lambda}=T_\beta+\lambda V_\omega, \quad \lambda>0 \ee
acting on  $\ell^2(\mathbb Z^d)$. The kinetic energy operator $T_\beta$, with $\beta > 0$, is described by a polynomially-decaying (or, power-law) kernel $T_\beta (n,m) \sim \|n-m\|^{-(d + 2 \beta  )}$. At strong disorder, the deterministic spectrum of this model is dense pure point spectrum and its eigenfunctions decay at least polynomially. This was proved by Aizenman and Molchanov \cite{Aizenman-Molchanov93} using the FMM, with an improvement by Disertori, Escobar, and Rojas-Molina \cite{DERM24}, see Section \ref{sec:long_range1}. For these models, which include the \textit{fractional Laplacian}, we prove, in strong disorder regime, that the local eigenvalue statistics (LES) is a  Poisson point process,
that the system exhibits dynamical delocalization throughout its spectrum, and that the global eigenfunction correlators are bounded below by a power-law expression.

Long-range hopping models, including the fractional Laplacian, have attracted recent interest in the spectral theory and quantum dynamics community as they appear in the modeling of anomalous diffusion, large connected networks, and random walks with long jumps. 
Long-range hopping operators perturbed by quasiperiodic potentials, have received increased attention in the last years \cite{Chulaevsky-Dinaburg93, Haro-Puig13, Jito-Liu21, Ge-You-Zhao22, Liu-Powell-Wang24, Ge-Jito24}. Here, $T_\beta$ appear as the dual operator of quasiperiodic potentials, through Aubry duality \cite{AA} (see \cite[section 9.10]{df_v2} for a textbook description). The latter, together with the lrAM (and fAM) belong to the class of operators with dynamically defined potentials.

The lrAM was first studied by mathematicians using multi-scale analysis (MSA) or the fractional moment method (FMM). When $T_\beta$ in \eqref{op-h-beta} is replaced by a long-range hopping operator $T(m,n)$ with \textit{exponential decay}, Klein \cite{klein} used the MSA to prove that at large disorder, the spectrum is pure point throughout its deterministic spectrum, with exponentially decaying eigenfunctions, almost surely, Aizenman and Molchanov \cite{Aizenman-Molchanov93} applied their FMM at large disorder and proved that for lrAM, with power-law decay of $T_\beta$,  
polynomial (or, power-law) localization holds: the eigenfunctions decay at least polynomially and the spectrum is almost surely pure point. This result was later improved upon in \cite{DERM24}, see also Section \ref{sec:long_range1}. 

Other papers obtained results with slower-than-exponential decay of the matrix elements $T(m,n)$. 
Shi \cite{Shi21} studied the lrAM at large disorder with polynomial decay as in \eqref{op-h-beta}, but with a decay exponent much larger than $d$,  
and showed polynomial localization at large disorder using MSA. 
Jian and Sun \cite{jian_sun_PAMS22} studied a similar model, also with a decay exponent much larger than $d$, prove polynomial localization at strong disorder by modifying the SULE eigenfunction condition for polynomially-decaying eigenfunctions. Additionally, they prove that small moments of the position operator 
are bounded, roughly consistent with our result, Theorem \ref{thm-no_dl}. For additional detail, see Remark \ref{rem:yo_conj1}. (We use the more common definition of dynamical localization that \textit{all the moments} are bounded in time.)
For similar power-law localization results in the context of one-photon states in disordered media, see \cite{Kraisler-Schenker-Schotland26}. 

Another group of researchers studied the lrAM for which the hopping term  $T (m, n)$ is \textit{subexponential} (but not polynomial) of the form 
$$
T(m, n) \sim e^{- \zeta (\log ( \|m-n\|) + 1))^\rho } , ~~~
 \zeta > 0, \rho > 1, 
$$
Shi and Wen \cite{shi_wen22} studied the subexponential case with a quasi-periodic potential using KAM techniques and proved polynomial localization. More recently, Jian, Xu, and Yan \cite{JXY} used the FMM to prove dynamical localization and the corresponding eigenfunction decay at the same rate, at large disorder. Shi, Wen and Yan \cite{Shi-Wen-Yan25} obtained similar results using MSA. For power-law localization results regarding localization under the presence of an electric field we refer to \cite{Aloisio26,
Aloisio-Oliveira-Pigossi26} and \cite{Aloisio-Matosetal26}, for  structural results regarding  localization centers and Semi-Uniform Localization in the power-law setting.

In \cite{Jaksic-Molchanov99}, Jaksic and Molchanov  showed (for any $\lambda > 0$) that in dimension $d=1$ and $\beta>3/2$, the spectrum of $H_{\beta,\lambda}$ has no absolutely continuous component, while if $\beta>7/2$, the spectrum has no continuous component. In dimension one, $H_{\omega,\lambda}$ can be seen as a bi-infinite Toeplitz (also called Laurent) matrix with diagonal disorder, exhibiting Lifshitz tails, a feature often associated with localization, see \cite{GRM22}.

The long-range random operator $H_{\omega,\lambda}$ has been intensively studied in the physics literature, see \cite{Cressoni-Lyra98,Rodriguez-et-al03}, where
numerical studies showed the existence of a transition in the eigenvalue statistics and eigenfunction localization, depending on the region of the spectrum the disorder parameter $\lambda$, and the parameter $\beta$ appearing in \eqref{eq:lr_kernel1-1}. The idea that long-range hopping may cause delocalization appeared in \cite{BurinMaksimov1989}, which sometimes motivates the term Burin-Maksimov model for $H_{\omega,\lambda}$. One of our results, Theorem \ref{thm-no_dl}, establishes a quantitative version of this delocalization phenomena brought on by the long-range hopping term $T_\beta$. 

\subsection{Basic definitions and assumptions}

The lrAM studied in the paper consists of a long-range hopping term and a random potential. These are assumed to satisfy the following hypotheses.

\begin{hypothesis}\label{hypothesis:lr1}
The operator $T_\beta$, for $\beta > 0$,  is a bounded, symmetric, translationally invariant, long-range hopping operator on $\ell^2(\Z^d)$  defined with a bounded kernel $T_\beta(n,m)$ so that  
\be\label{eq:lr_defn1-1} 
(T_\beta f )(m) :=    \sum_{n\in\mathbb Z^d,n\neq m} T_\beta (m,n) f(n) - C_{T_{\beta}} f(m),  
 \ee
with the constant $0 < C_{T_\beta} < \infty$ given by
\beq\label{eq:lr_constant1}
C_{T_\beta} :=  \sum_{n\in\mathbb Z^d, n\neq 0} T_\beta (0,n),
\eeq
which is a convergent series for every $\beta>0$, under Hypothesis \ref{hypothesis:lr1} 
Here, $T_\beta (0,0) = T_\beta (m,m)$ for all $m\in\mathbb Z^d$, is a finite constant, and the off-diagonal terms $T_\beta(n,m)$ decay polynomially, for $\beta>0$, as
\beq\label{eq:lr_kernel1-1}
\frac{c_{\beta,d}}{\|m-n\|^{d+2\beta}}\leq | T_\beta (m,n) | \leq \frac{C_{\beta,d}}{\|m-n\|^{d+2\beta}},  ~~~ m\neq n,
\eeq
for some constants $0 < c_{\beta,d} \leq C_{d, \beta} < \infty$. 
\end{hypothesis}

\medskip

With respect to the random potential, we make the following assumption.

\begin{hypothesis}\label{hypothesis:pot1}
The potential $V_\omega$ is an Anderson-type random potential. That is, it is given by 
\begin{equation}\label{rand-pot}
    V_\omega  = \sum_{k \in \Z^d} \omega_k P_k , 
\end{equation}
where $P_k$ is the rank-one projection onto site $k \in \Z^d$, that is, $P_k = \angles{\delta_k,\cdot}\delta_k$, with $\{\delta_k \}_{k\in \mathbb Z^d}$ the canonical orthonormal base of $\ell^2(\mathbb Z^d)$. The set 
$\{\omega_k \}_{k\in \mathbb{Z}^d}$ is a collection of independent, identically distributed (iid) random variables with a common probability measure $\mu$. We denote the probability space $\Omega= (\supp \mu) ^{\mathbb Z^d}$, with probability measure $\mathbb P=\bigotimes_{k\in\mathbb Z^d} \mu$ and expectation denoted by $\mathbb E(\cdot)$. We assume that $\mu$ has a bounded and compactly supported density $\rho \geq 0$ with $\int \rho = 1$. 
\end{hypothesis}

Under Hypotheses \ref{hypothesis:lr1} and \ref{hypothesis:pot1}, the lrAM $H_{\beta,\lambda}$ is an ergodic family  (see, e.\ g. \cite{AW15}). In this case, its almost-sure spectrum $\Sigma_\lambda$ is given by
\begin{equation}
    \Sigma_\lambda = \sigma (T_\beta) + \lambda ~\supp \mu.
\end{equation}
where $\sigma(T_\beta)$ is the spectrum of $T_\beta$. The spectrum of $T_\beta$ is a closed interval determined by the range of its dispersion relation $T_\beta (\theta)$, for $\theta \in \mathbb{T}^d$. For example, in the case where
$$
T_\beta (m,n) = \frac{C}{\|m-n\|^{d + 2\beta}},
$$
the dispersion relation $T_\beta(\theta)$ is 
$$
T_\beta(\theta) = C_\beta ~ \sum_{n \in \mathbb{Z}^d \backslash \{0\}} \left( \frac{\sum_{j=1}^d \cos(n_j \theta_j) }{\|n\|^{d + 2 \beta}} \right) ,  ~~~ \theta \in \mathbb{T}^d ,
$$
for an appropriate finite constant $C_\beta > 0$ (see, for example, \cite{Cressoni-Lyra98}).

The fractional Laplacian is a specific example of a long-range operator satisfying Hypothesis \ref{hypothesis:lr1}.  This was proven in \cite{Ciaurrietal18}  for the case $d=1$ and in \cite{GRM20} for general dimension.  
Let $-\Delta$ be  the discrete, nonnegative, finite difference Laplacian on $\Z^d$ with spectrum $\sigma(-\Delta) := [0, 4d]$. 
The fractional Laplacian $(-\Delta)^\alpha$ is defined through the spectral theorem.
We recall that, using the discrete Fourier transform, for all $\varphi\in \ell^2(\mathbb Z^d)$,
\be\label{eq:fourier-rep}
((-\Delta)^\alpha \varphi)(k) = \frac{1}{(2 \pi)^\frac{d}{2}} \int_{T^d} e^{i \overline{k} \cdot \overline{\theta}} \left(  \sum_{n=1}^d (2- 2 \cos \theta_n ) \right)^\alpha \widehat{\varphi}( \overline{\theta}  ) ~d \overline{\theta}, 
\ee
where $\hat \varphi$ is the Fourier transform of $\varphi$. It follows that the spectrum of the fractional Laplacian is $\sigma((-\Delta)^\alpha)  := [0, (4d)^\alpha]$. Taking $\alpha=1$ in \eqref{eq:fourier-rep} yields the usual Anderson model.

\begin{rem} In the probability literature one often writes $(-\Delta)^{\alpha/2}$ with $\alpha\in(0,2)$. In this case, the decay in inequality \eqref{eq:lr_kernel1-1} has exponent $d+\alpha$.
\end{rem}

The \emph{fractional Anderson Model} (fAM),  is given by
\be \label{def:frac-am} H_{\alpha,\lambda}=(-\Delta)^{\alpha}+\lambda V_\omega,\quad \alpha \in (0,1),\, \lambda>0 \ee
acting on $\ell^2(\mathbb Z^d)$, where $(-\Delta)^\alpha $ is the fractional discrete Laplacian with exponent $\alpha$ and $V_\omega$ is an Anderson potential satisfying Hypothesis \ref{hypothesis:pot1}

It was shown in \cite{Ciaurrietal18} in $d=1$ and in \cite{GRM20} in $d\geq 1$, that $(-\Delta)^\alpha$ verifies \eqref{eq:lr_kernel1-1} with $\beta=\alpha\in (0,1)$ for some constants $c_\beta=c_{\alpha,d}, C_\beta=C_{\alpha,d}$ depending on $\alpha$ and the dimension $d$, see also the recent \cite{Hake-Keller-Pogorzelski-26a,Hake-Keller-Pogorzelski-26b}. This makes the fractional Anderson model a particular case of the long-range random model described \eqref{op-h-beta}.
Under Hypothesis \ref{hypothesis:pot1}, the fAM $H_{\alpha,\lambda}$ is an ergodic family  (see, e.\ g. \cite{AW15}) with almost-sure spectrum given by  $\Sigma_{\alpha,\lambda}= [0, (4d)^\alpha] + \lambda ~\supp \mu$. In an abuse of notation, we will write $\Sigma_\lambda$ for $\Sigma_{\alpha,\lambda}$ when no confusion arises.

Since the fAM can be seen as a particular case of the lrAM, the proof of dense pure point spectrum for large values of $\lambda$, from \cite{Aizenman-Molchanov93}, still applies. However, in  \cite{Padgett_Liaw_etal19} it was observed numerically that the fAM exhibits less localization, in a suitably defined sense, than the usual Anderson model. In \cite{DERM24}, the authors proved the decay of the Green's function of the fAM using a self-avoiding walk expansion, following ideas of  \cite{schenker2009}. They improved the fractional Green's function estimates, and among others, obtained boundedness of the first moment of the time-evolved wave packets. The authors of \cite{DERM24} conjectured that  higher moments are not bounded and that, in particular, the fAM does not exhibit dynamical localization. This is confirmed by our result, Theorem \ref{thm-poisson}.

\subsection{Background: Delocalization, local eigenvalue statistics, eigenfunction correlators}\label{subsec:back1}

One of the main motivations to study the fAM comes from an expected metal-insulator transition depending on the dimension and on the exponent $\alpha$. We recall that for the usual Anderson model, it is expected that pure point spectrum with exponentially decaying eigenfunctions holds in $d=1$, for $d=2$, there are conflicting views, while in $d=3$ there should be an energy transition between localized states, for energies in the spectral edges, to delocalized states, associated to energies in the bulk of the spectrum, provided the disorder isn't too strong. Complete localization is known for the case $d=1$ and band-edge localization is known for $d=2$. The existence of a transition in $d=3$ remains a challenging open problem. 

In the fractional Anderson model, however, this transition is expected to take place in dimension $d=1$, as observed numerically in \cite{Cressoni-Lyra98,Rodriguez-et-al03}. This is motivated by the recurrence-transience transition known for the random walk with long jumps associated to $(-\Delta)^\alpha$ in any dimension: the walk is known to be transient for $\alpha<d/2$ and recurrent for $\alpha\geq d/2$, see e.g. \cite{Caputo-Faggionato-Gaudilliere09}. In the one-dimensional case, therefore, we expect  different spectral and dynamical behaviors for the fAM for $\alpha<1/2$ and $\alpha\geq 1/2$. A phase transition depending on the exponent $\alpha$ and the dimension is known in the case of the long-range random Ising model, see e.g. \cite{DingHuangMaia2024,AffonsoBissacotMaia2023}
 and references therein.

Going back to the perspective of general long-range random models, in Theorems \ref{thm-poisson} and \ref{thm-no_dl}, below, we show that $H_{\omega,\lambda}$ exhibits Poisson eigenvalue statistics and dynamical delocalization for large $\lambda$ and $\beta>d/2$. This implies, in particular, the same results for the fractional Anderson model with $\alpha\in (1/2,1)$ and $d=1$. The co-existence of Poisson eigenvalue statistics and dynamical delocalization is strikingly different from the apparent dichotomy in the eigenvalue statistics of the Anderson model and other random quantum systems described above, although certain strongly correlated random systems are known to exhibit some transport along with localization in a certain direction \cite{Mavi-Schenker19,{Matos-Mavi-Schenker24}}.

For Anderson-type models with i.i.d.\ potentials, eigenvalue level repulsion is associated with extended eigenstates and delocalization, 
as found in matrices where all entries are random, while Poisson eigenvalue statistics is associated to exponentially localized eigenstates, and dynamical localization, found in the Anderson model, where disorder is only on the diagonal terms. This association is also found in the case of random band matrices and heavy-tailed random matrices. In the former case, it has been shown that in a conveniently defined limit where the matrix size grows and the width of the band decreases, the eigenvalue statistics become Poisson distributed and the associated eigenvectors become exponentially localized, deviating from the usual random matrix case and approaching the Anderson model behavior, see \cite {Casati-et-al1990,Fyodorov-Mirlin1994}, and \cite{hislop-krishna2007, schenker2009, chen-smart2023,cipolloni-et-al2025}. In the case of heavy-tailed random matrices, extreme eigenvalues are known to exhibit Poisson eigenvalue statistics with associated eigenvectors that are exponentially decaying \cite{Kieburg-Monteleone21,{Auffinger-BenArous-Peche09}}. 

One might ask whether Poisson eigenvalue statistics is a structural consequence of having exponentially decaying eigenfunctions, which in turn is a consequence of the exponential decay of the Green's function of the operator. Our first main result, Theorem \ref{thm-poisson}, shows that this is not the case and that Poisson statistics can be obtained with non-exponential decay of the Green's function. To our knowledge, this is the first example of a model defined on $\ell^2(\mathbb Z^d)$, with diagonal disorder, exhibiting this behavior. At finite volume, and rescaling the hopping terms with matrix size and strength of disorder, it was observed numerically in \cite{Tang-Khaymovich22} that the model exhibited both Poisson eigenvalue statistics and non-exponentially decaying eigenfunctions.

For the model in \eqref{op-h-beta}, Yeung and Oono conjectured in \cite{Yeung-Oono-87}, based on numerical evidence, that the eigenfunctions do not decay exponentially and instead inherit the decay in \eqref{eq:lr_kernel1-1}. That is, the eigenfunctions have the same decay rate as the kernel of $T_\beta$. Upper bounds showing (at least) polynomial decay have been obtained, under certain conditions, with the standards proof of localization \cite{Aizenman-Molchanov93,CMS90,Jaksic-Molchanov99, GRM20, Shi21,DERM24,Shi-Wen-Yan25,JXY}, giving a partial proof of the conjecture, while obtaining (averaged) lower bounds remains an open problem. Our results on dynamical delocalization do not involve decay of eigenfunctions, but, rather, are based upon lower bounds on the Green's function fractional moments. Using these bounds, we show that time-averaged moments of the position operator in states evolving under the action of $e^{-itH_{\omega,\lambda}}$ are infinite, for all large moments, see Theorem \ref{thm-no_dl}.  In order to study power-law, lower bounds on eigenfunctions, we establish lower bounds on the global eigenfunction correlators as stated in Corollary \ref{cor:ef-corr}. 
These lower bounds suggest that a suitable average of the eigenfunctions satisfy lower bounds of the form $C \|n\|^{- (d+ 2 \beta)}$, for some $C > 0$ and $n \neq 0$, see Remark \ref{rem:yo_conj1}. 
We believe that our results are a step closer to a complete proof of Yeung and Oono's conjecture.

\subsection{Main results: Local eigenvalue statistics and absence of dynamical localization} \label{subsec:fr_results1}

In order to study the local eigenvalue statistics (LES) of an ergodic random operator $H$ on $\ell^2(\mathbb Z^d)$, we consider finite-volume restrictions of the operator and study the distribution of the rescaled eigenvalues near a fixed energy $E$ in the almost-sure spectrum $\Sigma$. Let $\Lambda \subset \Z^d$ be a cube of volume $|\Lambda|$ and denote by $1_\Lambda$ the characteristic function on $\Lambda$. The operator $H^\Lambda := 1_\Lambda H 1_\Lambda$ acts on $\ell^2 (\Lambda)$ and has discrete spectrum consisting of (random) eigenvalues $\{ E_j(\lambda) \}_{j=1}^{|\Lambda|}$. We define the point process on $\R$ by
\be\label{def:pp-xi} 
\xi(\Lambda, E)=\sum_{j=1}^{|\Lambda|} \delta_{\abs{\Lambda}(E_j(\Lambda)-E)}, 
\ee 
and study the limit as $|\Lambda| \to \infty$. 

In the case of the lattice Anderson model $H_\omega=-\Delta+\lambda V_\omega$, Minami showed, under certain conditions, that the process $\xi(\Lambda,E)$ converges to a Poisson point process with intensity function $n(E) := N'(E)$, where $N(E)$ is the integrated density of states of $H_\omega$, provided $n(E) > 0$, which occurs for almost every $E \in \Sigma_\lambda$. We say the Anderson model exhibits Poisson local eigenvalue statistics \cite{Minami96}.  Minami's proof is based on estimates on the probability of finding eigenvalues in a given interval, the Wegner estimate and the Minami estimate, and the exponential decay of the Green's function of $H_\omega$. The Wegner and Minami estimates are a consequence of the regularity of the probability distribution $\rho$ in Hypothesis \ref{hypothesis:pot1} and the rank-one structure of $V_\omega$. The decay of Green's function is given by the Fractional Moment Method for large disorder.
These ingredients are crucial in controlling the error terms in the truncation procedure described above. 

In the case of long-range operators, the non-locality of $T_\beta$, and the consequent power-law, rather than exponential, decay of the Green's function, makes the truncation errors much more difficult to control. In Theorem \ref{thm:gen-minami}, we show that the condition on exponential decay of the Green's function in Minami's proof can be weakened to allow for only polynomial decay. We do this by exploiting the improved decay bounds on the Green's function obtained in \cite{DERM24}, and by pushing the summability of the errors in the truncation method, see Section \ref{sec:poisson1}.  This allows us to prove that  \eqref{def:pp-xi} converges to a Poisson point process, with the correct intensity, thereby establishing the existence of Poisson eigenvalue statistics at almost every $E \in \Sigma_\lambda$.  

In this article, we will always work in the strong localization regime so that the deterministic spectrum is purely pure point and the eigenfunctions decay polynomially.

\begin{thm}\label{thm-poisson}
We consider either of the following settings,
\begin{itemize}
    \item[i)] Let $d\geq 1$ and $H_\lambda=H_{\beta,\lambda}$ be the long-range Anderson Model given in \eqref{op-h-beta} satisfying Hypotheses \ref{hypothesis:lr1} and \ref{hypothesis:pot1}, with $\beta>d/2$. 
    \item[ii)] Let $d=1$ and $H_\lambda=H_{\alpha,\lambda}$ be the fractional Anderson model given in \eqref{def:frac-am} satisfying Hypothesis \ref{hypothesis:pot1} with $\alpha\in (\frac{1}{2},1)$.
\end{itemize}
Then, there exists $\lambda_0>0$, depending on the parameters of the model, such that for all $\lambda>\lambda_0$, for almost every $E \in \Sigma_\lambda$, the point process $\xi(\Lambda, E)$ converges weakly as $\abs{\Lambda} \to \infty$ to a Poisson point process with intensity $n(E)dx$, where $n(E)>0$ is the density of states function of $H_\lambda$.
\end{thm}

Technically, we recall that the assumption of large $\lambda$ in Theorem \ref{thm-poisson} corresponds to the region where the Green's function of the operator decays polynomially.
This was proved in  \cite{DERM24} for the fAM, with a decay in the Green's function of the form $|x|^{-s(d+2\alpha)}$, which is stronger than the one obtained from the Fractional Moment Method, which is of the form $|x|^{-2s\alpha}$. We generalize this result to the lrAM in Theorem \ref{thm:spectral-loc-lram}, with improved bounds.
This improvement of an extra $d$ term in the decay exponent of the Green's function allows us to treat the error bounds, with a convenient summation argument that pushes the limits of Minami's proof.

In the Anderson model, the assumption of $\lambda > 0$ large corresponds to the strong disorder regime, where, in addition to exponential localization, dynamical localization holds throughout the deterministic spectrum (under some regularity conditions). 
We recall the following definition from \cite[Definition 2.4]{GKduke} in a version suitable for our setting.

\begin{defn}[Dynamical (de)localization]\label{def:dyn-loc}
\textit{We say that an ergodic operator $H_{\omega}$ on $\ell^2(\mathbb Z^d)$, parametrized by $\omega\in\Omega$ with $\Omega$ a probability space, exhibits \textbf{dynamical localization} in its almost-sure spectrum if}
\be\label{dynloc}
\mathbb E \left( \sup_{t\in\mathbb R} \norm{\angles{X}^{q/2} e^{-itH_{\omega}}P_0}^2   \right) <\infty , \quad \mbox{for all}\, ~~ q\geq 0 ,
\ee
\textit{where $\mathbb E(\cdot)$ denotes the expectation in $\Omega$ and $P_0$ is the rank-one projection onto the site $0\in\mathbb Z^d$. If there exists $q$ for which \eqref{dynloc} does not hold, we say the operator exhibits \textbf{dynamical delocalization.}}
\end{defn}
The expression on the left-hand-side of \eqref{dynloc} is called the $q$-th moment of the position operator. 
Note that dynamical localization, that is, the fact that all moments of the position operator are bounded, does not follow directly from the polynomial decay of the Green's function for fAM and lrAM in Theorems \ref{thm:decay-fractional-GF} and \ref{thm:decay-long-range-GF}. Despite showing that the fAM exhibits pure point spectrum and that the first moment of the position operator is bounded, in \cite[Remark after Theorem 2]{DERM24}, the authors \emph{conjecture that the fractional Anderson model does not exhibit dynamical localization}. This is in stark contrast with the usual Anderson model, where both features are present. In fact, it will follow from our analysis that, if $\beta>d/2$, in the strong disorder regime of Theorems \ref{thm:decay-fractional-GF} and \ref{thm:decay-long-range-GF}, \textbf{dynamical localization does not hold for the lrAM.} This is a consequence of the long-range property of the operator $T_\beta$.  Our method of proof uses a self-avoiding random walk expansion for the Green's function, and the lower bound in \eqref{eq:lr_kernel1-1} to show that the dominant term in the Green's function expansion is given by the shortest path, which in the non-local scenario of the lrAM, is a long jump. This yields our second main result.

\begin{thm}\label{thm-no_dl}
We consider any of the following settings,
\begin{itemize}
    \item[i)] Let $d\geq 1$ and $H_\lambda=H_{\beta,\lambda}$ be the long-range Anderson model given in \eqref{op-h-beta} satisfying Hypotheses \ref{hypothesis:lr1} and \ref{hypothesis:pot1}, with $\beta>d/2$. 
    \item[ii)] Let $d=1$ and $H_\lambda=H_{\alpha,\lambda}$ be the fractional Anderson model given in \eqref{def:frac-am} satisfying Hypothesis \ref{hypothesis:pot1} with $\alpha\in (\frac{1}{2},1)$.
\end{itemize}
Then, there exists $\lambda_1\geq \lambda_0$, depending on the parameters of the model,  such that for all $\lambda>\lambda_1$, $H_\lambda$ does not exhibit dynamical localization in its almost-sure spectrum $\Sigma_\lambda$.
\end{thm}

\begin{rem}

\begin{enumerate} 
\item In Hypothesis \ref{hypothesis:pot1}, the assumption of $\rho$ having compact support is made for convenience.  Theorem \ref{thm-no_dl} holds more generally under a suitable assumption on the moments of $\rho$ if it has  noncompact support. 

\item In our proof, we present a lower bound on the expectation of the fractional moments of the Green's function which holds in arbitrary complex vicinities of $\Sigma_\lambda$, see Theorem \ref{thm:lb_green1} below. This lower bound might be of independent interest. 

\item As Theorem \ref{thm-no_dl} covers a regime where the random operator is usually considered to be strongly localized, we expect that similar delocalization phenomena should hold for more general classes of operators possessing a power-law decaying kernel.
\end{enumerate}
\end{rem}

To prove Theorem \ref{thm-no_dl}, we consider the time-averaged moments of the $q$-th moments of the wave packets, given in \eqref{defn:moment1}, Section \ref{sec:deloc1}. Note that dynamical localization implies that the time-averaged $q$-th moments are also bounded for all $q\geq 0$. We  show in  Theorem \ref{thm:no_dl1} that, under the conditions of Theorem \ref{thm-no_dl}, all time-averaged $q$-th moments of the position operator with $q\geq d + 4 \beta$ in case (i) and $q\geq d+4\alpha$ in case (ii), are infinite. This yields Theorem \ref{thm-no_dl}.

For comparison, in
\cite[Section 7]{Kerner-Post-Sabri-Taufer25}, the authors study the behavior of the first moment of the position operator for the unperturbed fractional Laplacian $(-\Delta)^\alpha$ for $d=1$. This operator has only absolutely continuous spectrum. The authors show that the first moment of the position operator grows with $t$ (ballistic behavior) $\alpha> 1/4$, while it is infinite for all $t$ if $\alpha\leq 1/4$ (super-ballistic behavior). Theorem \ref{thm-no_dl}-case (ii) shows that for $\alpha \in (\frac{1}{2}, 1)$, the addition of a random potential creates superballistic motion for moments with $q \geq 1 +  4 \alpha$.

\subsection{Outline}

The remainder of this article is structured as follows: in the next Section we recall an upper bound on the decay of the Green's function for the fAM obtained in \cite{DERM24} and generalize it to operators whose hopping terms exhibit slow inverse polynomial decay, as the lrAM. We summarize the consequences of these upper bounds on spectral localization and decay of eigenfunctions. In Section \ref{sec:poisson1} we show a generalization of a result by Minami \cite{Minami96} on Poisson local eigenvalue statistics and apply it to our setting. In Section \ref{sec::lower1} we obtain lower bounds for the Green's function of operators whose hopping terms do not decay faster than an inverse polynomial, including fAM and lrAM. In Section \ref{sec:deloc1} we use the lower bounds on the Green's function to show that the time-averaged moments are unbounded for large disorder, and to obtain lower bounds on the eigenfunction correlations.



\section{Upper bounds on the Green's function for long-range random operators}\label{sec:long_range1}

In this section we recall results on the decay of the Green's function for the fractional Anderson Model and generalize them to the long-range Anderson Model.

We first recall the estimate obtained in \cite[Theorem 2 (i)]{DERM24} on the decay of the averaged fractional power of the Green's function for the fractional Anderson Model  $H_{\alpha,\lambda}$. While the original result is stated for the operator on the full space, the proof is also valid for any finite-volume restriction (see \cite[Section 4]{DERM24}).
Let $\Lambda=\Lambda_L \subset \mathbb Z^d$, a cube of side length $L$, and consider the finite-volume restriction $H^\Lambda$. We write $G^\Lambda(m,n;z)= \langle \delta_m, (H^{\Lambda}-z)^{-1} \delta_n \rangle$, for $m,n \in \Lambda_L$.

\begin{thm}[\cite{DERM24}] \label{thm:decay-fractional-GF} Let $H_{\alpha,\lambda}=(-\Delta)^\alpha+\lambda V_\omega$ be the fAM with $\alpha\in (0,1)$, with a random potential satisfying Hypothesis \ref{hypothesis:pot1}. For $s\in(\frac{d}{d+2\alpha},1)$, there exists  $\lambda_0(\alpha,s)>0$ such that for $\lambda>\lambda_0(\alpha,s)$ and all $m,n \in\mathbb Z^d$, $m \neq n$, we have
\be \label{GF-decay0}
\mathbb E\left(  \abs{G^\Lambda(m,n;z)}^s \right) \leq \frac{{C_{\lambda,d,\alpha,s}} }{\norm{m-n}^{s(d+2\alpha)}} 
\ee
uniformly in $z\in\mathbb C\setminus \mathbb R$, where $C_{\lambda,d,\alpha,s}$ is a finite constant depending on $\lambda,d,\alpha$ and $s$.
\end{thm}

We recall an important consequence of the above result on the spectral localization for fAM at strong disorder,  \cite[Theorem 2 (ii)]{DERM24}, see also \cite[Corollary 2.4]{GRM20} :
\begin{thm}[\cite{DERM24},\cite{GRM20}] \label{thm:spectral-loc-fam}
 Let $H_{\alpha,\lambda}$ be as Theorem \ref{thm:decay-fractional-GF}. For $s\in(\frac{d}{d+2\alpha},1)$, there exists  $\lambda_0(\alpha,s)>0$ such that for $\lambda>\lambda_0(\alpha,s)$ the spectrum $\Sigma_\lambda$ of $H_{\alpha,\lambda}$ consists only of pure point spectrum with polynomially decaying eigenfunctions for a.e.\ $\omega\in\Omega$. 
\end{thm}




The proof of \cite[Theorem 2 (i)]{DERM24} relies on the connection between the decay of the fractional moment of the Green's function to the two-point correlation function of a self-avoiding random walk (SAW) with long jumps, where the transition probabilities between any two different sites is related to the hopping terms in the operator $H_{\alpha,\lambda}$, see \cite[Theorem 1]{DERM24}. In general, a SAW with long jumps is characterized by the transition probabilities $D(n,m)$ between two different sites $n,m\in\mathbb Z^d$, where the function $D:\mathbb Z^d\times \mathbb Z^d\rightarrow \mathbb R$ is non-negative, symmetric, translation invariant, and such that
\be M:=\sum_{n\neq 0} D(0,n) <\infty.
\ee
In the case of the long-range Anderson Model, the long-range hopping terms determine the transition probabilities of a SAW with long jumps, with a function $D$
by 
\be D(n,m)=\frac{C}{\norm{n-m}^{d+2\beta}},\quad \beta>0\ee for $m\neq n$, with $D(m,m)$ uniformly bounded in $m\in\mathbb Z^d$. With this representation, \cite[Theorem 1]{DERM24} can be readily generalized to long-range Anderson models with arbitrary $\beta>0$, as the associated two-point correlation function $C^D_\gamma(n)=\sum_{m\geq 0}c_m^D(n)\gamma^m$ is well defined for $\gamma<R$ where the radius of convergence $R$ satisfies $R>\frac{1}{M}$, see \cite[Eq. (3.7)]{DERM24}.
 Note that the proof is also valid for any finite-volume restriction to a cube $\Lambda$. In this setting, \cite[Lemma 2.4]{CS15} on the decay of the two-point correlation function of a SAW with long jumps still applies and yields the following improved bound. 

\begin{thm}\label{thm:decay-long-range-GF} Let $H_{\beta,\lambda}$ be the long-range random operator satisfying Hypotheses \ref{hypothesis:lr1} and \ref{hypothesis:pot1}. For $s\in(\frac{d}{d+2\beta},1)$ and $\beta > 0$, there exists  $\lambda_0(\beta,s)>0$ such that for $\lambda>\lambda_0(\beta,s)$ all $m,n\in\mathbb Z^d$, $n\neq m$, we have
\be \label{GF-decay0-LR}
\mathbb E\left(  \abs{G^\Lambda(m,n;z)}^s \right) \leq \frac{C_{\lambda,d,\beta,s}}{\norm{n-m}^{s(d+2\beta)}} ,
\ee
uniformly in $z\in\mathbb C\setminus \mathbb R$, where $C$ is a finite constant depending on $\lambda,d,\beta$ and $s$.
\end{thm}

Note that, just as in \cite{DERM24}, the use of a SAW representation gives an improvement on the upper bound on the fractional moment of the Green's function obtained in \cite[Theorem 3.1]{Aizenman-Molchanov93}. Following the proof of \cite[Theorem 2]{DERM24}, we obtain the results on the spectral localization for the long-range Anderson Model at large disorder and a stronger decay estimate on the eigenfunctions. 


\begin{thm} \label{thm:spectral-loc-lram}
Let $H_{\beta,\lambda}$ be the long-range random operator satisfying Hypotheses \ref{hypothesis:lr1} and \ref{hypothesis:pot1}. For $s\in(\frac{2d}{d+2\beta},1)$, and $\beta > \frac{d}{2}$, there exists  $\lambda_0(\beta,s)>0$ such that for $\lambda>\lambda_0(\beta,s)$ the following holds:
\begin{itemize}
    \item[i.] The spectrum $\Sigma_\lambda$ of $H_{\beta,\lambda}$ consists only of pure point spectrum for a.e. $\omega\in\Omega$. 

    \item[ii.] The eigenfunctions $\varphi_{k,\omega}$ associated to the eigenvalues of $H_{\beta,\lambda}$ satisfy, for some $n_{k,\omega}\in\mathbb Z^d$, for some constant $C_{k,\omega,n_{k,\omega},s}>0$, all $n \neq n_{k,\omega}$, and $\frac{2d}{d + 2 \beta} < s'<s < 1$ and $ s - s' > \frac{d}{d+2\beta}$,
 \be\label{eq:aeupperboundphi}
 \abs{\varphi_{k,\omega}(n)} \leq \frac{C_{k,\omega,n_{k,\omega},s}}{\norm{n-n_{k,\omega}}^{s'(d+2\beta )}} ,
\ee
almost surely. Thus, $\abs{\varphi_{k,\omega}(n)}$ decays (at least) at a rate arbitrarily close to $2\beta \approx d$.  

\item[iii.] For $0<q< 2\beta$, for a.e. $\omega\in\Omega$, we have
\be \sum_{n\in\mathbb Z^d\setminus\{0\}} \|n\|^q \abs{e^{-itH_{\beta,\lambda}}(0,n)}^2 < \infty.
\ee
\end{itemize}
\end{thm}
\begin{rem}\label{rem:bounded-moments-lram}
Results (i) and (iii) hold for $s \in ( \frac{d}{d+ 2 \beta} , 1 )$ and for $\beta > 0$. 
  Note that Theorem \ref{thm:spectral-loc-lram} (iii) implies the boundedness of the first $\floor{2\beta}$ moments of a wave function initially localized at the origin and evolving under the action of $H_{\beta,\lambda}$, with $\beta>0$. 
\end{rem}

As mentioned, Yueng and Oono \cite{Yeung-Oono-87} conjectured that the eigenfunctions should decay at the same rate as the hopping term $T_\beta$. That is, one expects that a typical eigenfunction $\varphi_{k,\omega}(x)$ is bounded above as
\beq\label{eq:ef_opt1}
|\varphi_{k,\omega}(m)|  \leq \frac{C_{k,\omega}}{\|m\|^{s(d + 2 \beta)}} .
\eeq
For $s$ close to one, this rate of decay is approximately $d + 2\beta \sim 2d$. Hence, the almost sure power-law upper bound in part (ii) of Theorem \ref{thm:spectral-loc-lram} is not expected to be optimal for most configurations $\omega$. Indeed, we prove a better rate of eigenfunction decay but only with a good probability $t_n$, with $t_n \to 1$, as $|n| \to \infty$. 
We sketch the proof of this improved bound 
from the Green's function \eqref{GF-decay0-LR} and eigenfunction correlator bounds, see section \ref{subsec:lb_ec} for the definitions. The local eigenfunction correlator $Q(m,n;I)$ for an interval $I \subset \Sigma_\lambda$ is bounded above by the Green's function. It is proved in \cite[Theorem 7.7]{AW15} that 
\beq\label{eq:ef_corr1}
\E(Q(m,n;I)) \leq C_s  \liminf_{|\eta| \to 0} \int_I \E ( |G (m,n;E+i \eta)|^s) ~dE. 
\eeq
Using \eqref{GF-decay0-LR} and Markov's inequality, one sees that for each $n \neq m$ and any $t_n>0$, 
\beq\label{eq:correlator_bd1}
\mathbb{P}(Q(m,n;I)>t_n)\leq \frac{1}{t_n} \left( \frac{C_{\lambda,d,\beta,s}|I|}{\norm{m-n}^{s(d+2\beta)}} \right).
\eeq
In particular, we fix $m$, take any $\varepsilon>0$, and set $t_n=\frac{C_{\lambda,d,\beta,s}|I|}{\norm{m-n}^{s(d+2\beta)-\varepsilon}}$. It follows from \eqref{eq:correlator_bd1} that  
\beq\label{eq:correlator_bd2}
\mathbb{P} \left( Q(m,n;I)>\frac{C_{\lambda,d,\beta,s}|I|}{\norm{m-n}^{s(d+2\beta)-\varepsilon}} \right) \leq \frac{1}{\norm{m-n}^{\varepsilon}}.
\eeq
 From \eqref{eq:correlator_bd2}, we derive two conclusions on eigenfunction upper bounds:
 \begin{itemize}
     \item[i)] Large probability result.  For any $\beta>0$ and any$\varepsilon>0$:
     $$\lim_{n\to \infty}\mathbb{P} \left( Q(m,n;I)\leq \frac{C_{\lambda,d,\beta,s}|I|}{\norm{m-n}^{s(d+2\beta)-\varepsilon}} \right) = 1, 
     $$
     which yields a probabilistic version of the upper bound in Yeung and Oono's conjecture \cite{Yeung-Oono-87}.
     
     \item[ii)] Almost sure result. For $\beta>d/2$, we choose $s',s\in (\frac{2d}{2+2\beta},1)$, with $s'<s$. Setting $\varepsilon=(s-s')(d+2\beta)$ and requiring that $\varepsilon > d$, we see from \eqref{eq:correlator_bd2} that
     $$
     \sum_{n\in \mathbb{Z}^d\setminus\{m\}}\mathbb{P}(Q(m,n;I)>t_n)<\infty .
     $$
     An application of the Borel-Cantelli Lemma yields the existence of a full measure set $\Omega_m$ such that  $\omega \in \Omega_m$ implies
     \begin{equation}{\label{eq:correpoitaeb}} Q(m,n;I)\leq \frac{C_{\lambda,d,\beta,s,\omega}|I|}{\norm{m-n}^{s'(d+2\beta)}}.
     \end{equation}
     Let $\Omega'=\cap_{m\in \mathbb{Z}^d} \Omega_m$. Then $\Omega'$ is a full measure set such that \eqref{eq:correpoitaeb} hold for all $\omega\in \Omega'$ and all $m\in \mathbb{Z}^d$.
     Upon expanding the correlator in terms of eigenfunctions, we have
{}
the bound \eqref{eq:aeupperboundphi} by picking a center of localization $m$ for $\varphi_k$. }
\end{itemize}

\begin{rem}\label{rem:yo_conj1} We compare this eigenfunction decay result with prior bounds for eigenfunctions of the lrAM with power-law hopping.

\begin{enumerate}
    \item  Aizenman and Molchanov \cite[section 3]{Aizenman-Molchanov93}
proved that for kernels $T_\beta$ satisfying Hypothesis \ref{hypothesis:lr1}, the eigenfunctions are bounded above by a power-law as in \eqref{eq:ef_opt1} but with exponent $[s (d + 2 \beta) - d] \gtrsim d$, almost surely, similar to the result in part (ii) of Theorem \ref{thm:spectral-loc-lram}. This follows from their FMM and an application of the Simon-Wolff criterion. 

\item A similar almost sure power-law exponent of $\frac{1}{2}[s (d + 2 \beta) - d ] \gtrsim \frac{d}{2}$ was proven in \cite{DERM24} for the frAM.  

\item Shi \cite{Shi21} studied long-range hopping terms with a power-law decay exponent $r > d$ ($r = d + 2 \beta$ in our notation). For $r > 331d$, Shi \cite[Theorem 2.5]{Shi21} proved an almost sure bound on the decay of the eigenfunctions with exponent $\frac{r}{600}$ using MSA.  

\item  Jian and Sun \cite{jian_sun_PAMS22} improved Shi's result obtaining a SULE-type estimate involving the centers of localization and a power-law decay exponent in $[0, \frac{r}{160}]$, almost surely, when $r > 331d$. 

\end{enumerate}
\end{rem}


\section{Poisson local eigenvalue statistics for the case of long-range hopping terms} 
\label{sec:poisson1}

 In this section, we prove a generalization of the theorem in \cite{Minami96} on local eigenvalue statistics for the Anderson model on $\ell^2(\Z^d)$ to the long range setting. Minani's result states that the point process $\xi(\Lambda, E)$, given in \eqref{def:pp-xi}, converges to a Poisson point process provided $E$ is in the localized region of the deterministic spectrum. We prove an analogous result for the lrAM in this section establishing Theorem \ref{thm-poisson}.

 We recall that for a given self-adjoint operator $H$, $H^X$ denotes its restriction to the set $X\subseteq \mathbb Z^d$ with simple boundary conditions, that is, $H^X=1_{X}H1_{X}$ as an operator on $\ell^2(X)$. We will often consider the cube of side-length $L$, $\Lambda_L=[-\ceil{\frac{L}{2}},\ceil{\frac{L}{2}}]^d\cap\mathbb Z^d$ and for simplicity, write $H^L$ instead of $H^{\Lambda_L}$. The matrix elements of $H$ are denoted $H(n,m)$. The Green's function of $H^X$ is denoted $G^X(z)$ with matrix elements $G^X(n,m;z)$

\begin{thm}\label{thm:gen-minami}
Let $H_{\beta,\lambda}$ be a random operator of the form $H_{\beta,\lambda}=T_\beta +\lambda V_\omega$ with $T_\beta$ verifying, for some constants $\beta>0$, $C_T>0$,
\be\label{eq:decay-T}  \abs{T_\beta(n,m)}\leq \frac{C_T}{\norm{n-m}^{d+2\beta}} ,\quad \mbox{for all }n\neq m \in\mathbb Z^d\ee
with  $\lambda >0$ and $V_\omega$ satisfying Hypothesis \ref{hypothesis:pot1} with a bounded probability density $\rho$. Denote by $\Sigma_{\beta,\lambda}\subset  \mathbb R$ its almost sure spectrum, and assume that for $E\in \Sigma_{\beta,\lambda }$ the following holds,
\begin{itemize}
\item[(i)] the density of states function for $H_{\beta,\lambda}$, denoted by $n(E)$, exists and is positive $n(E) > 0$ at $E$. 
\item[(ii)] the Green's function $G_{\Lambda_L}$ of the restriction of $H_{\beta,\lambda}$ to a cube $\Lambda_L$ satisfies the following decay estimate: there are finite constants $s\in (\frac{2d}{d + 2 \beta}, 1)$, $C>0$, and $r>0$,  such that
\be \label{GF-decay}
\mathbb E\left(  \abs{G^{\Lambda_L}(n,m;z)}^s\right) \leq \frac{C}{\norm{n-m}^{s(d+2\beta)}}  ,
\ee
for all cubes $\Lambda_L \subseteq \mathbb Z^d$ and all $n,m\in\Lambda_L$, with $n\neq m$ and $z$ with $Im(z)\neq 0$ such that $\abs{z-E}<r$. 
\end{itemize}
Then, if $\beta > \frac{d}{2}$, the point process $\xi(\Lambda_L, E)$ converges weakly as $L \to \infty$ to a Poisson point process with intensity measure $n(E)dx$.
\end{thm}
\begin{rem}
The above Theorem shows that for random Scrh\"odinger operators with pure point spectrum, with random potentials regular enough to satisfy the aforementioned hypotheses, the hopping terms $|T(n,m)|$ must decay at least faster than $\norm{n-m}^{-2d}$, in order to exhibit Poisson eigenvalue statistics. Note also that the case of exponentially decaying hopping terms and short-range hopping terms, like the usual Anderson Model, are included in the setting of Theorem \ref{thm:gen-minami}. 
\end{rem}

\smallskip

Combining the above generalization of Minami's work, Theorem \ref{thm:gen-minami}, together with the results from Section \ref{sec:long_range1}, we have:

\begin{proof}[Proof of Theorem \ref{thm-poisson}]
 Let $H_{\beta,\lambda}$ be an operator satisfying Hypothesis \ref{hypothesis:lr1} and Hypothesis \ref{hypothesis:pot1}, which in particular imply \eqref{eq:decay-T}. The regularity of the random potential ensures that a Wegner estimate holds, see Lemma \ref{lem:wegner-minami} below, which by \cite[Corollary 5.24]{kirsch} implies that the integrated density of states $N(E)$ is absolutely continuous with bounded density given by $n(E)$. This implies that  $N$ is differentiable almost everywhere in the spectrum, which yields $n(E)>0$ almost everywhere. Therefore, condition (i) is verified a.e.\ in $\Sigma_\lambda$.  Moreover, for the long-range Anderson model,  Theorem \ref{thm:decay-long-range-GF} ensures the existence of $\lambda_{\text{lrAM}}>0$ such that \eqref{GF-decay} holds for all $\lambda> \lambda_{\text{lrAM}}$ with $\beta>d/2$. The analogue statement holds for the fractional Anderson model by Theorem \ref{thm:decay-fractional-GF}, for all $\lambda>\lambda_{\text{fAM}}$ for some $\lambda_{\text{fAM}}$ and $\beta=\alpha\in(\frac{1}{2},1)$. Therefore, both the long-range Anderson model and the fractional Anderson model verify the conditions in Theorem  \ref{thm:gen-minami}, which yields the desired result.
\end{proof}

The proof of Theorem \ref{thm:gen-minami} follows the strategy of  Minami's argument in \cite{Minami96} but it requires additional technical estimates due to the long-range nature of the hopping term. 
In preparation for the proof of Theorem \ref{thm:gen-minami}, we recall the strategy of Minami's argument for the Anderson model in \cite{Minami96}. Then, in section \ref{subsec:conv1}, we present the proof of the infinite-divisibility of the limiting point process, and in section \ref{subsec:proof31} we complete the proof of Theorem \ref{thm:gen-minami}.

Let $\xi(E)$ be a Poisson point process on $\mathbb R$. Under the assumptions of the Theorem, the weak convergence of $\xi(\Lambda_L,E)$ to a limiting point process as $L\rightarrow \infty$ is equivalent to the convergence of certain functionals of the point processes on a space $\mathcal A$ of tests functions (\cite[Lemma 1]{Minami96}). The space $\mathcal A$ consists of finite linear combinations of functions of the form 
\be \label{def:test-functions}
f_\zeta(x)=\frac{b}{(x-a)^2+b^2} 
\ee
with $\zeta=a+ib$, $a\in\mathbb R,b>0$, also called Poisson kernels. Note that, by functional calculus (see \cite[Eq. (2.18)]{Minami96}
\be\label{rep-process-resolvent} \mathbb E \left( \xi(\Lambda_L,E)(f_z)\right) = \frac{1}{L^d} \mathbb E\left(\mbox{Im}\, \mbox{Tr } (H_{\beta,\lambda}^L-E-\frac{\zeta}{L^d})^{-1} \right). \ee

This representation allows us to reduce the analysis of weak convergence of point processes to a study of convergence of resolvents of the operator $H_{\beta,\lambda}$ restricted to cubes.

Throughout the proof, two key estimates are used: the Wegner estimate and the Minami estimate, concerning the number of eigenvalues of the operator $H_{\beta,\lambda}^L$ in an interval $\mathcal I$, that is, $\mbox{Tr} \,1_{\mathcal I}( H_{\beta,\lambda}^L )$.
The Wegner and Minami estimates for lrAM follow most directly from the proofs in Combes, Germinet, and Klein \cite{cgk_JSP} since these models involve rank-one perturbations and the single-site probability measure has a density.

\begin{lem}\label{lem:wegner-minami}
The lrAM operator  $H_{\beta,\lambda}$, with $T_\beta$ satisfying Hypothesis \ref{hypothesis:lr1} and with a random potential $V_\omega$ satisfying Hypothesis \ref{hypothesis:pot1}, exhibits a Wegner estimate and a Minami estimate. There exist constants $C_W, C_M > 0$, depending on the probability density $\rho$ and the disorder $\lambda$, so that  for any interval $\mathcal{I} \in \Sigma_{\beta,\lambda}$, we have
\beq\label{eq:wegner}
 \mathbb E \left( \mbox{Tr} \,1_{\mathcal I}( H_{\beta,\lambda} ^L) \right)\leq C_W |\mathcal{I}| L^d ,
 \eeq 
 and
 \be \label{est:minami} \mathbb E \left(\mbox{Tr} \,1_{\mathcal I}( H_{\beta,\lambda}^L ) \left( \mbox{Tr} \,1_{\mathcal I}( H_{\beta,\lambda}^L )-1 \right)   \right)\leq C_M |\mathcal{I}|^2 L^{2d} .
 \ee 
\end{lem}
By the Chebychev inequality, these bounds on the expectations imply the probability estimates:
\be \mathbb P \left(  \mbox{Tr} \,1_{\mathcal I}( H_{\beta,\lambda}^L ) \geq 1\right) \leq  C_W |\mathcal{I}|   L^d ,\ee
 and
 \be \mathbb P \left(    \mbox{Tr} \,1_{\mathcal I}( H_{\beta,\lambda}^L ) \geq 2 \right)\leq  C_M |\mathcal{I}|^2  L^{2d} 
 \ee

Having the representation \eqref{rep-process-resolvent}, and the Wegner and Minami estimates, the proof that the limiting process $\xi(E)$ is a Poisson point process $n(E)dx$ follows in two steps:

\smallskip

\smallskip

\noindent {\bf Step 1.} The point process $\xi(\Lambda_L,E)$  can be approximated by a superposition of smaller independent point processes, that is, the limiting point process $\xi$, if it exists, is infinitely divisible. 
\smallskip

\noindent {\bf Step 2.} The smaller independent processes in Step 1 are asymptotically negligible in a certain sense to be defined below, and their superposition converges weakly to $\xi$, the Poisson point process with intensity $n(E)dx$. Together with Step 1, this yields the proof of Theorem \ref{thm:gen-minami}.

\smallskip

\smallskip

We devote the next subsection to prove Step 1. Thereafter, we complete the proof of Theorem \ref{thm:gen-minami} by proving Step 2. The proofs are based on  the use of the geometric resolvent identity to compare resolvents on different cubes, and the use of the decay of the Green's function to control the error terms. These manipulations differ from those in \ref{thm:gen-minami}, due to the fact that we need to exploit better the decay of the resolvents to control the error terms.


\subsection{Local decoupling and infinitely-divisibility}
\label{subsec:conv1}

Following the notation in \cite{Minami96}: 
Given $L>0$, we decompose the cube $\Lambda_L$ into $N_L^d$ smaller cubes $\mathcal C_p$ with side length scale $\ell := L / N_L \ll L$, to be chosen later. We have
\be\label{eq:lambda_decomp1}
\Lambda_L  = {\rm Int} \bigcup_{p=1}^{N_L^d} \overline{\mathcal C_p} .
\ee
Relative to $\mathcal C_p$, we have the decomposition 
\be 
H^{\Lambda_L}=H^{\Lambda_L \setminus \mathcal C_p}\oplus H^{\mathcal C_p}+K,
\ee
where $H^{\mathcal C_p}=1_{\mathcal C_p}H^{\Lambda_L}1_{\mathcal C_p} $, and the kernel $K(m,n)$ represents the boundary contribution of the long-range hopping terms. The kernel satisfies $K(m,n)\neq 0$ only if $m\in \mathcal C_p$ and $n\in\Lambda\setminus \mathcal C_p$ or the other way around, and for such a pair $(m,n)$, we have the bound  
\be\label{eq:kernelK1}
 \abs{K(m,n)}\leq \frac{C_T}{\norm{m-n}^{d+2\beta}}, ~~~ \beta > 0,
\ee
where $C_T$ is the constant coming from \eqref{eq:decay-T}. 

For each $p$, we define the point process,
\be \label{def:eta-p}
\eta(\mathcal C_p,E)=\displaystyle \sum_{j=1,...,|\mathcal C_p|}\delta_{| \Lambda_L|(E_j(\mathcal C_p)-E)},
\ee
where $E_j(\mathcal C_p)$ denotes the eigenvalues of $H_{\mathcal C_p}$, and consider their superposition 
\be \eta(\Lambda_L,E)=\displaystyle \sum_{p=1}^{N_L^d} \eta(\mathcal C_p,E).  \label{eq:eta_pp1}
\ee

Then, to show that the point process $\xi(\Lambda_L,E)$ is approximated by the process $\eta(\Lambda_L,E)$ it is enough to show that this approximation holds on test functions of the form \eqref{def:test-functions}. That is, that for all $\zeta=a+ib$ with $a\in\mathbb R$, $b>0$,
\be \lim_{L\rightarrow \infty}\abs{ \mathbb E \left( \xi(\Lambda_L,E)(f_\zeta)\right)- \mathbb E \left( \eta(\Lambda_L,E)(f_\zeta)\right) }=0.  \ee
This follows from the representation \eqref{rep-process-resolvent} and the following result, where we use the notation  $R_{X}(z)=(H_{\beta,\lambda}^X-z)^{-1}$, for the resolvents of $H_{\beta,\lambda}$ restricted to $X$, for  $X=\Lambda_L,\mathcal C_p$.
\begin{prop}\label{prop:resolv_diff1}
	 Under the hypotheses of Theorem \ref{thm:gen-minami}, with $\beta>d/2$ and $s\in (\frac{2d}{d+2\beta},1)$, we have, for $\zeta=a+ib$, $a\in\mathbb R$, $b>0$,
	\be\label{eq:resovl_diff1}
	\lim_{L\rightarrow \infty}
	\frac{1}{L^d} \left[ \Im {\rm Tr} R_{\Lambda_L} (z_{E,L}) - \sum_{p=1}^{N_L^d} \Im {\rm Tr} R_{\mathcal C_p}(z_{E,L})  \right] =0 ,
	\ee
where $z_{E,L}=E+\frac{\zeta}{L^d}$, and ${\rm Tr}$ denotes the trace on $\Lambda_L$ or $\mathcal C_p$ according to the argument.
\end{prop}

\begin{proof}
 Following the decomposition in \eqref{eq:lambda_decomp1}, for each cube $\mathcal C_p$ on scale $\ell$, we designate a security zone in $\mathcal C_p$ near its boundary, as follows: 
\be
{\rm Int} \mathcal C_p := \{ m \in \mathcal C_p ~|~ \dist ( m, \partial \mathcal C_p ) > \ell_L \},
\ee
for a scale $\ell_L$ satisfying $0 < \ell_L \ll \ell \ll L$ to be chosen below.  

Recalling \eqref{eq:lambda_decomp1}, we expand the trace in the square brackets in \eqref{eq:resovl_diff1} and obtain:
\bea\label{eq:decomp1}
 \lefteqn{ \frac{1}{L^d} \left[ \Im {\rm Tr} R_{\Lambda_L} (z_{E,L}) - \sum_{p=1}^{N_L^d} \Im {\rm Tr} R_{\mathcal C_p}(z_{E,L}) \right]}
& & \nonumber  \\
  &= & \frac{1}{L^d} \sum_{p=1}^{N_L^d} \sum_{m \in \mathcal C_p} \left[ \Im G^{\Lambda_L} (m,m;z_{E,L}) - \Im G^{\mathcal C_p}(m,m;z_{E,L}) \right] \nonumber \\
  & =: & A_L + B_L,
  \eea
  where
 \be\label{eq:ALdefn1}
 A_L : =   \frac{1}{L^d} \sum_{p=1}^{N_L^d} \sum_{m \in \mathcal C_p \backslash {\rm Int} \mathcal C_p } \left[ \Im G^{\Lambda_L} (m,m;z_{E,L}) - \Im  G^{\mathcal C_p}(m,m;z_{E,L}) \right],
 \ee
 and
 \be\label{eq:BLdefn1}
  B_L : =   \frac{1}{L^d} \sum_{p=1}^{N_L^d} \sum_{m \in {\rm Int} \mathcal C_p } \left[ \Im G^{\Lambda_L} (m,m;z_{E,L}) - \Im G^{\mathcal C_p}(m,m;z_{E,L}) \right] . 
 \ee

 \noindent
 1. \textbf{Bound on $\E \{ A_L \}$}. First note that $\E \left( \Im  G^{X}(m,m;z_{E,L}) \right)$ is uniformly bounded in $L$, for $X = \mathcal C_p$ and $X = \Lambda_L$. The proof follows as in \cite[Eq. (2.22)]{Minami96}, yielding 
 \be\label{eq:bd-c-rho} \mathbb E \left(  \Im  G^{X}(m,m;z_{E,L})  \right)\leq \pi \frac{\norm{\rho}_\infty}{\lambda} :=C_{\rho,\lambda} \ee

  To simplify the notation, when there is no ambiguity, we write $z=z_{E,L}$.   Recall that the Green's function $G$ (and its finite-volume versions) satisfies the Herglotz property, that is, $\Im G(m,m;z)/ \mbox{Im} z >0$, for $\mbox{Im} z \neq 0$, therefore $\Im G(m,m;z) \geq 0$, and the terms are homogeneous with respect to $p$. Therefore, we have from \eqref{eq:ALdefn1}, 
 \bea
 \E \{A_L \} & \leq &    \frac{1}{L^d} \sum_{p=1}^{N_L^d} \sum_{m \in \mathcal C_p \backslash {\rm Int}\mathcal C_p } \left[ \E \{ \Im G^{\Lambda_L} (m,m;z) \} + \E \{ \Im  G^{\mathcal C_p}(m,m;z) \} \right]   \nonumber \\
  & \leq &  \frac{2C_{\rho,\lambda} }{L^d} \, N_L^d \, | \mathcal C_1 \backslash {\rm Int} \mathcal C_1 | \nonumber \\
   & \leq &  \frac{2C_{\rho,\lambda} }{L^d} \,  N_L^d \, \left( \frac{L}{N_L} \right)^{d-1} \ell_L \nonumber \\
    & \leq & 2C_{\rho,\lambda}\frac{N_L \ell_L}{L}.
 \eea
 
 Since $\ell_L<< \ell=L/N_L$, we have that the right-hand side of the last line goes to zero as $L$ goes to infinity, therefore
 \be \lim_{L\rightarrow\infty}\E \{A_L \}= 0  
 \ee
 
\smallskip

\noindent 
2. \textbf{Bound on $\E \{ B_L \}$.} We next estimate the expectation of \eqref{eq:BLdefn1}, by considering $\E \{ |B_L|^{s/2} \}$. We use the fact that $\left(\sum_k a_k\right)^{s/2}\leq \sum_k a_k^{s/2}$, for any sequence of non-negative numbers $a_k$ and $s\in(0,1)$.  Taking the expectation, and using the geometric resolvent identity (see \cite[Section 5.3]{kirsch}), we obtain
\bea\label{eq:BLest1}
\E \{ | B_L|^\frac{s}{2} \} & \leq &   \frac{1}{L^\frac{ds}{2}} \sum_{p=1}^{N_L^d} \sum_{m \in {\rm Int} \mathcal C_p } 
 \mathbb E\left(|G^{\Lambda_L} (m,m;z) - G^{\mathcal C_p}(m,m;z) |^\frac{s}{2} \right)  \nonumber \\
 & \leq &   \frac{1}{L^\frac{ds}{2}} \sum_{p=1}^{N_L^d} \sum_{m \in {\rm Int} \mathcal C_p }  
 \sum_{\substack{k \in \Lambda_L \backslash \mathcal C_p \\   n \in \mathcal C_p}} \mathbb E\left(  | G^{\Lambda_L} (m,k;{z}) K (k,n) G^{\mathcal C_p} (n,m;z) |^{s/2} \right) \nonumber \\
 & \leq &   \frac{C_T^{s/2}}{L^\frac{ds}{2}} \sum_{p=1}^{N_L^d} \sum_{m \in {\rm Int} \mathcal C_p }  
 \sum_{\substack{k \in \Lambda_L \backslash \mathcal C_p \\   n \in \mathcal C_p}}  \frac{ \E \left(| G^{\Lambda_L} (m,k;{z})|^{s} \right)^\frac{1}{2} \E \left( |G^{\mathcal C_p} (n,m;z) |^s \right)^\frac{1}{2} }{ \| k - n \|^{\frac{s}{2} (d + 2 \beta)}} ,
 \nonumber \\
  & & 
  \eea
  where we used \eqref{eq:kernelK1} and the Cauchy-Schwarz inequality for the expectation. This is the analog of \cite[(2.35)]{Minami96}. 
By hypothesis, we have
\be\label{eq:green_L1} 
 \mathbb E \left(\abs{G^{\Lambda_L}(m,k;z)}^{s} \right) \leq \frac{C}{\norm{m-k}^{s (d+2\beta)}},
\ee
for some positive constant $C$.
By construction, for $m\in ~\rm{Int} \mathcal C_p$ and $k\in \Lambda_L \backslash \mathcal C_p$ we have $\norm{m-k}> \ell_L$, therefore
\be\label{est:G}  \mathbb E \left( \abs{G^{\Lambda_L}( m,k;z)}^{s} \right)^{\frac{1}{2}} \leq \frac{C}{\ell_L^{\frac{s}{2}(d+2\beta)}}, \quad \mbox{for} \, m\in ~\rm{Int} \mathcal C_p,\,k\in \Lambda_L \backslash \mathcal C_p.
\ee
Noting that these bounds are uniform with respect to $p$, we obtain from  \eqref{eq:BLest1}, 
\bea\label{eq:BLest2}
 \E( |B_L|^{s/2})
  &\leq & \frac{C^{(1)}}{L^{\frac{ds}{2}}} \frac{1}{\ell_L^{\frac{s}{2}(d+2\beta)}}\sum_{p=1}^{N_L^d}  \sum_{m\in \rm{Int} \mathcal C_p} ~\sum_{ \substack{n\in \mathcal C_p \\ k\in \Lambda\setminus \mathcal C_p}}\frac{\mathbb{E}\left( |G^{\mathcal C_p} (n,m;z) |^s \right)^\frac{1}{2}}{\norm{n-k}^{\frac{s}{2}(d+2\beta)}} \nonumber \\
   &\leq & \frac{C^{(1)} \, N_L^d}{L^{\frac{ds}{2}} \ell_L^{\frac{s}{2}(d+2\beta)}}   \sum_{\substack{n\in \mathcal C_1 \\ k\in \Lambda\setminus \mathcal C_1}} \frac{1}{\norm{n-k}^{\frac{s}{2}(d+2\beta)}}  \left( \sum_{m\in \rm{Int} \mathcal C_1}\mathbb{E}\left( |G^{\mathcal C_1} (n,m;z) |^s \right)^\frac{1}{2}  \right)  \nonumber \\
    \eea
with $C^{(1)}=C_T^{s/2}C^{1/2}$. We consider the sum over $m \in \rm{Int} \mathcal C_1$ in \eqref{eq:BLest2} with $n \in \mathcal C_1$ fixed.   We divide the sum over $m \neq n$ and $m=n$. Note that, by hypothesis, the expectation of $|G^{\mathcal C_1}(n,m;z)|^s$ satisfies the decay in \eqref{GF-decay} for $n\neq m$, and is uniformly bounded for $n=m$, that is,
\begin{equation}\label{eq_green_cp1}
 \mathbb{E}\left( |G^{\mathcal C_1} (n,m;z) |^s \right) \leq C\left\{ \begin{array}{cc} 
 \frac{1}{\norm{n-m}^{s (d+2\beta)}}, & n \neq m; \\
 1 , \ n = m ,
  \end{array} \right.
\end{equation}
The uniform bound for $n=m$ is the so-called \textit{a priori} bound, which is known to hold under our hypotheses, see e.g. \cite[Lemma 4.1]{stolz} whose proof also holds in our setting, or \cite[Eq. (3.2)]{DERM24}). It is known that the constant $C$ in this estimate is  $C=\frac{C_{\rho,s}}{\lambda^s}$, for some constant $C_{\rho,s}$ depending only on $\rho$ and $s$. Then, using \eqref{eq_green_cp1}, the sum over $m \neq n$ with fixed $n \in \mathcal C_1$ yields 
\beq\label{eq_green_cp2}
\sum_{m\in \rm{int} \mathcal C_1\setminus\{n\}}\mathbb{E}\left( |G^{\mathcal C_1} (n,m;z) |^s \right)^{\frac{1}{2}}\leq \sum_{k\in \rm{Int} \mathcal C_1\setminus\{n\}}\frac{1}{\norm{n-k}^{d+2\kappa}}=:C^{(2)}  < \infty ,
\eeq
where we write 
\be\label{eq:kappa} \frac{s}{2}(d+2\beta)=d+2\kappa, \,\mbox{with}\,\, \kappa>0 \ee
which is possible since $\beta>d/2$ and $\frac{2d}{d+2\beta}<s<1$.
 The term $m =n$ can be bounded using the \textit{a priori} bound. This yields, using the notation \eqref{eq:kappa},
\be\label{eq:sum}   \sum_{\substack{n\in \mathcal C_1 \\ k\in \Lambda\setminus \mathcal C_1}} \frac{1}{\norm{n-k}^{\frac{s}{2}(d+2\beta)}}  \left( \sum_{m\in \rm{Int} \mathcal C_1}\mathbb{E}\left( |G^{\mathcal C_1} (n,m;z) |^s \right)^\frac{1}{2}  \right) \leq C^{(3)} \sum_{\substack{n\in \mathcal C_1 \\ k\in \Lambda\setminus \mathcal C_1}} \frac{1}{\norm{n-k}^{d+2\kappa}}, 
\ee
where $C^{(3)}=C+C^{(2)}$. Substituting this bound into \eqref{eq:BLest2} gives
\begin{equation}\label{eq:newBLest23}
\mathbb E( |B_L|^{s/2}) \leq \frac{C^{(4)} N_L^d}{L^{\frac{ds}{2}} \, \ell_L^{ d+2\kappa }}  \sum_{\substack{n\in \mathcal C_1 \\ k\in \Lambda\setminus \mathcal C_1}}\frac{1}{\norm{n-k}^{d+2\kappa}},
\end{equation}
with $C^{(4)}=C^{(1)}C^{(3)}$. In order to control the double sum on the right-hand side, we recall a result obtained in \cite[Eq. (3.3)-(3.9)]{GRM20}. 
\begin{lem}[\cite{GRM20}]\label{Lemma-Gebert-RM} 
Let $\Lambda_K=[-K,K]\cap\mathbb Z^d$ be a cube, and $\kappa>0$, then
\begin{itemize}
\item[(i)] If $2\kappa \in (0,1)$, then
\be\label{lemma-Gebert-RM-a}  \sum_{n\in\Lambda_K}\sum_{k\in\Lambda^c _K} \frac 1 {|n-k|^{d+2\kappa}} \leq \left( K+\frac{1}{2} \right)^{d-2\kappa}  
\ee 
\item[(ii)] If $2\kappa\geq 1$, then, for all $\epsilon\in (0,1)$,
\be\label{lemma-Gebert-RM-b}  \sum_{n\in\Lambda_K}\sum_{k\in\Lambda^c _K} \frac 1 {|n-k|^{d+2\kappa}} \leq \left( K+\frac{1}{2} \right)^{d-(1-\epsilon)}  
\ee
\end{itemize}
\end{lem}

\noindent
Note that in \cite{GRM20} this result is stated with $\epsilon=1/2$.

Then, to estimate the sum in \eqref{eq:newBLest23}  we can apply Lemma \ref{Lemma-Gebert-RM} with $K=\ceil{\ell/2}>0$ , yielding
$$
\sum_{n\in \mathcal C_1,k\in \Lambda\setminus \mathcal C_1}\frac{1}{\norm{n-k}^{d+2\kappa}}\leq  \ell^{d-2\gamma}.
$$
where $\gamma=\kappa$, if $2\kappa<1$, whereas if $2\kappa\geq 1$ then $\gamma<1/2$ and $\gamma$ may be taken arbitrarily close to $1/2$.
Then \eqref{eq:newBLest23} reduces to
\begin{align}\label{eq:newBLest24}
\mathbb E( |B_L|^{s/2}) & \leq   \frac{C^{(4)} N_L^d}{L^{\frac{ds}{2}} \, \ell_L^{ d+2\kappa }} 
\ell^{d-2\gamma}  \nonumber \\ 
& = C^{(4)} \frac{ L^{d\left( 1-\frac{s}{2} \right)} }{\ell_L^{d+2\kappa} \ell^{2\gamma}} \\
\end{align}
 where we used the fact that $N_L=\frac{L}{\ell}$. Next, we set
\be  \label{eq:theta} \ell=L^{\theta} \,\mbox {and } \, \, \ell_L=L^{\theta'} \,\, \mbox{with}\,\, 0<\theta'<\theta<1. \ee
This yields
\be\label{eq:final-b_L} \mathbb E( |B_L|^{s/2})  \leq C^{(4)} L^{d \left( 1-\frac{s}{2} \right)- d\theta'- 2\kappa \theta' -2\gamma \theta   } . \ee
Taking $\theta'>1-\frac{s}{2}$, which is possible since $s<1$, and noting that $\gamma,\kappa,\theta>0$, we obtain that  $\mathbb E( |B_L|^{s/2}) \to 0$, as $L \to \infty$.  
Together with the bound on $\mathbb E(A_L)$, these imply that $A_L + B_L \to 0$, as $L \to \infty$, in probability, yielding \eqref{eq:resovl_diff1} .

\end{proof}

\subsection{ Proof of Theorem \ref{thm:gen-minami} }\label{subsec:proof31}

We prove that the point process $\eta(\Lambda_L,E)$, defined in \eqref{eq:eta_pp1}, converges to a Poisson process with intensity measure $n(E)dx$ as $L\rightarrow \infty$. This result, together with Proposition \ref{prop:resolv_diff1}, completes the proof of Theorem \ref{thm:gen-minami}.

First note that our hypotheses, together with Lemma \ref{lem:wegner-minami}, ensure that we can argue as in Equations (2.18)-(2.23) in \cite{Minami96} to show the existence of some constant $c_{\rho,\lambda}$ such that
\be  \mathbb E\left(\eta(\Lambda_L,E)(dx) \right) \leq c_{\rho,\lambda} dx,   \ee
and $n(E)dx\leq c_{\rho,\lambda} dx $, with  Therefore \cite[Lemma 1]{Minami96} can be applied. 
The locality of the random potential in Hypothesis \ref{hypothesis:pot1} implies that the point processes $(\eta(\mathcal C_p,E))_p$ are independent. Moreover, following a similar argument as in  Equations (2.18)-(2.23) in \cite{Minami96}, we have
\be \mathbb E\left( \eta(\mathcal C_p,E)(dx) \right)\leq \frac{1}{N_L^d} c_{\rho,\lambda}dx,  \ee
which implies $(\eta(\mathcal C_p,E))_p$ is a uniformly asymptotically negligible array (see \cite[Eq. 2.46]{Minami96}. Their superposition $\eta(\Lambda_L,E)$ converges weakly to a Poisson process with intensity measure $n(E)dx$ if the conditions in Equations (2.47) and (2.48), respectively, in \cite{Minami96} are verified. Condition \cite[Eq. (2.48)]{Minami96} is a consequence of the Minami estimate, see Lemma \ref{lem:wegner-minami}. It remains to prove condition \cite[Eq. (2.47)]{Minami96}, that is,
\be \label{cond:2.47-minami}  N_L^d \mathbb E\left( \eta(\mathcal C_p,E)(f_\zeta) \right) \rightarrow \pi n(E)  \ee 
for $f_\zeta$ given in \eqref{def:test-functions}, with $\zeta=a+ib$, $a\in\mathbb R$, $b>0$.
 \bigskip

 To show \eqref{cond:2.47-minami}, we first note that, arguing as in \eqref{rep-process-resolvent},
 \bea  \mathbb E(\eta(C_p,E)(f_\zeta)) & = &\frac{1}{L^d}\mathbb E \left( \mbox{Im\,}\mbox{Tr }G^{\mathcal C_p}(E+ \frac{\zeta}{L^d}) \right) \nonumber\\  
 & = &\frac{1}{L^d} \left(\displaystyle \sum_{m\in \rm{Int}C_p}+\sum_{m\in \mathcal C_p\setminus \rm{Int}\mathcal C_p}   \right)\mathbb E(\mbox{Im\,}G^{C_p}(m,m;z_{E,L})) \nonumber\\
 & = & D_L +F_L, 
 \eea
 where we recall that $z_{E,L}=E+\frac{\zeta}{L^d}$.  For  the term $D_L$, we write
\be G^{\mathcal C_p}(m,m;z_{E,L}) = G(m,m;z_{E,L})-\left( G(m,m;z_{E,L})-G^{\mathcal C_p}(m,m;z_{E,L})\right) \ee
to get, 
\bea \label{est:D_L} D_L & = &  \frac{1}{L^d} \sum_{m\in \rm{Int}\mathcal C_p}\mathbb E(\mbox{Im\,}G(m,m;z_{E,L})) - \mathcal E_L\nonumber \\
& =& \frac{\abs{ \rm{Int}\mathcal C_p}}{L^d}\,\mathbb E(\mbox{Im\,}G(m,m;z_{E,L})) - \mathcal E_L,
\eea
where 
\be \mathcal E_L= \frac{1}{L^d} \sum_{m\in \rm{Int}\mathcal C_p}  \mathbb E(\mbox{Im\,} \left\{ G(m,m;z_{E,L})-G^{\mathcal C_p}(m,m;z_{E,L})  \right\} ).
\ee
Therefore, we have
\be\label{est:eta-C_p} N_L^d \, \mathbb E(\eta(C_p,E)(f_\zeta)) = N_L^d\, \frac{\abs{ \rm{Int}\mathcal C_p}}{L^d}\mathbb E(\mbox{Im\,}G(m,m;z_{E,L}))  + N_L^d F_L- N_L^d \mathcal E_L. \ee
 To estimate $N_L^d F_L$, the uniform bound \eqref{eq:bd-c-rho} yields
 \be N_L^d\,F_L \label{est:F_L} \leq  N_L^d\, C_{\rho,\lambda} \frac{\ell^{d-1}\,\ell_L}{L^d}=  N_L^d\,C_{\rho,\lambda}  N_L^{-d} \frac{\ell_L}{\ell} = C_{\rho,\lambda} \frac{\ell_L}{\ell},\ee
which tends to zero as $L$ tends to infinity, since $\ell_L<\ell$. To estimate the last term on the right-hand side of \eqref{est:eta-C_p}, we write $N_L^d \mathcal E_L=\mathbb E (\tilde{\mathcal E_\ell} ) $ with
\be 
\tilde{\mathcal E_\ell}=  \frac{1}{\ell^d} \sum_{m\in \rm{Int}\mathcal C_p}  \mbox{Im\,} \left\{ G(m,m;z_{E,L})-G^{\mathcal C_p}(m,m;z_{E,L})  \right\}, \ee
Note that this expression is very similar to  $B_L$ in \eqref{eq:BLdefn1}, except $B_L$ carries an extra summation over $p$. We will estimate  it in the same way in order to prove that $\tilde{\mathcal E_\ell} \to 0$, as $L \to \infty$, almost surely. 
This will follow from proving that $\mathbb E( \tilde{|\mathcal E_\ell|}^{s/2}) \to 0$  as $L \to \infty$, almost surely. By the geometric resolvent identity, recalling \eqref{eq:kernelK1}, \eqref{eq:green_L1}, \eqref{est:G}, \eqref{eq:kappa}, and using Cauchy-Schwartz, we get
\begin{align}
\mathbb E( \tilde{|\mathcal E_\ell|}^{s/2}) & \leq \frac{1}{\ell^{d\frac{s}{2} } } \sum_{m\in \rm{Int} \mathcal C_p } \sum_{k\in \mathcal C_p^c,k'\in\mathcal C_p} \frac{C}{\norm{m-k}^{\frac{s}{2}(d+2\beta)}} \frac{\mathbb E\left( \abs{G^{\mathcal C_p}(k',m) }^{s/2} \right)}{ \norm{k-k'}^{ \frac{s}{2}(d+2\beta)} }   \nonumber \\
& \leq  \frac{C'}{\ell^{d\frac{s}{2}} \, \ell_L^{d+2\kappa}}  
\sum_{k\in \mathcal C_p^c,k'\in\mathcal C_p} \frac{1}{\norm{ k-k'  }^{d+2\kappa}} , 
\end{align}
where we used $2 \kappa = \frac{s}{2}(d + 2 \beta) - d > 0$, since $s > \frac{2d}{d + 2 \beta}$. The sum over $m \in \rm{Int}\mathcal C_p$ is bounded as 
in \eqref{eq_green_cp2}-\eqref{eq:sum}. Applying Lemma \ref{Lemma-Gebert-RM}, and recalling the choice \eqref{eq:theta}, yields
\be \label{est:E_L}
\mathbb E( \tilde{|\mathcal E_\ell|}^{s/2})  \leq \frac{C'}{\ell^{d\frac{s}{2}} \, \ell_L^{d+2\kappa}}  \ell^{d-2\gamma}= C' L^{ \theta d\left(1-\frac{s}{2}\right) - 2\gamma \theta -d\theta' - 2\kappa \theta'  }
\ee
where $2\gamma=2\kappa$ if $2\kappa\in(0,1)$ and $2\gamma=1-\varepsilon$ if $2\kappa\geq 1$, with $\varepsilon\in (0,1)$.
For this expression to tend to zero as $L$ grows, given that $\kappa,\gamma,\theta'>0$, it is enough to have
\be\label{eq:cond-theta}  \theta\left( 1- \frac{s}{2}\right) <\theta',  \ee
which is possible by taking $\theta,\theta'$ close enough to one depending on $s$, keeping the relation $0<\theta'<\theta$. To see that this is possible, write 
\be  \theta= 1-\frac{1}{n}, \quad \theta'=1-\frac{2}{n}, \quad n\in\mathbb N.
\ee
Then,  \eqref{eq:cond-theta} is a consequence of 
\be \frac{1}{n}\left( 1+\frac{s}{2} \right) < \frac{s}{2}, \ee
which holds for all $n>n_*$, for some $n_*\in\mathbb N$ large enough depending on $s$. We recall that the conditions $2 \kappa > 0$ and $s \in (\frac{2d}{d + 2 \beta}, 1)$ require $\beta > \frac{d}{2}$.

 Finally, combining \eqref{est:eta-C_p}, \eqref{est:F_L}, and \eqref{est:E_L}, and noting that $N_L^d\, \frac{\abs{ \rm{Int}\mathcal C_p}}{L^d} $ approaches $1$ as $L$ grows, yields
\be \lim_{L\rightarrow \infty}N_L \mathbb E(\eta(C_p,E)(f_\zeta)) = \pi \frac{dN(u)}{du}|_{u=E}=\pi n(E) \ee 
for all $E$ point of differentiability of $N$, proving \eqref{cond:2.47-minami}.

\qed


%
%
%
%
%


\section{Lower bounds for Green's functions fractional moments of long-range Anderson models}\label{sec::lower1}

In this section, we prove lower bounds on the expectation of fractional moments of the Green's function for the lrAM, including the fAM. As in previous sections, we let $H_{\beta, \lambda} =T_\beta +\lambda V_{\omega}$, where the long-range hopping operator $T_\beta$, for $\beta > 0$, satisfies the conditions of Hypothesis \ref{hypothesis:lr1} and the Anderson-type random potential $V_\omega$ satisfies the conditions of Hypothesis \ref{hypothesis:pot1}. We recall that $\Sigma_{\lambda}$ denotes the almost sure spectrum of $H_{\beta,\lambda}$.  For simplicity of notation, we write $T$ for $T_{\beta}$ and $\Sigma$ for $\Sigma_{\lambda}$. 
%
%
%
%
%

\subsection{\textit{A priori} lower bounds for Green's functions fractional moments}\label{subsec:apriori_lb1}

In the following, $\Lambda\subset \mathbb{Z}^d$ denotes a finite box. For $s\in (\frac{d}{d+2\beta},1)$ we define $\| T \|_s$ by
\beq\label{eq:T_defn1}
\|T\|^s_s:=\sum_{m}|T(0,m)|^s<\infty   .
\eeq
 As in previous applications of the fractional moment method \cite{Aizenman-Molchanov93,Aizenman,Schenkerl} we start with an \textit{a priori} bound on the expectation of fractional moments of the diagonal elements of the Green's function, with the key difference that here we aim at lower bounds. The result given below, of a non-local nature, is somewhat inspired by the proof of \cite[ Theorem 12.8]{AW15}.


\begin{lem}\label{Lem:generalapriori}
Let $s\in (\frac{d}{d+2\beta},1)$, $\lambda>0$ and $B > 0$ be given. There exists a positive constant $C_{\mathrm{AP}}=C_{\mathrm{AP}}(\lambda,B,\rho,s,\|T\|_s)$ such that whenever  $z\in \mathbb{C}^+\cap D(T(0,0),\lambda B)$ the following \emph{a priori} bounds hold for all $\Lambda\subset \mathbb{Z}^d$ and $k\in \Lambda$:
\beq\label{eq:genaprioribothdir}
C_{\mathrm{AP}}\leq \min\{ \mathbb{E}(|G_{\Lambda}(k,k;z)|^s),\mathbb{E}^{-1}\left(|G_{\Lambda}(k,k;z)|^{-s} \right )\}.
\eeq
Moreover we may take $C_{\mathrm{AP}}=\frac{1}{\lambda^s(M(s)+B^s)+\frac{\theta_s\|T\|^{2s}_{s}}{\lambda^s}}$, with
$M(s)=\mathbb{E}(|\omega_0|^s)$
and $\theta_s=\frac{4^s\|\rho\|^s_{\infty}}{1-s}$.

\end{lem}

\begin{proof}
1. Let $k=0$ for simplicity. For any $z \in \mathbb{C} \backslash \Sigma_\lambda$, the identity 
$\delta_0=(H^\Lambda -z)^{-1}(H^\Lambda -z)\delta_0$ implies
$$
1=(\lambda\omega_0+T(0,0) -z)G^\Lambda (0,0;z)+\sum_{n\neq 0}G^\Lambda (0,n;z)T(n,0).$$
The depleted resolvent identity \cite[(10.4)]{AW15} yields, for $n\neq 0$,
$$
G^\Lambda (0,n;z)= - G^\Lambda (0,0;z)\sum_{m\neq 0}T(0,m)G^{\{0\}^c}(m,n;z),$$
where we abbreviated $G^{\{0\}^c}(m,n;z):=G^{\Lambda\setminus\{0\}}(m,n;z)$.
Combining both equations above, we obtain
\begin{equation}\label{eq:Greeniden1}
1=G^\Lambda (0,0;z)\left(\lambda\omega_0+T(0,0)-z - \sum_{n\neq 0}\sum_{m\neq 0}T(0,m)G_{\{0\}^c}(m,n)T(n,0)\right).
\end{equation}
We let $J(0;z)$ denote the factor in parentheses on the right of \eqref{eq:Greeniden1}.  
Taking absolute values, raising both sides of \eqref{eq:Greeniden1} to the power $s/2$, where $s\in (\frac{d}{d+2\beta},1)$, taking expectations, and applying Cauchy-Schwarz, we have
\begin{equation}\label{eq:CS}
1\leq \mathbb{E}^{1/2}(|G(0,0;z)|^s)\mathbb{E}^{1/2}(|J(0;z)|^s)  ,  
\end{equation}
where
\begin{equation}\label{eqJbounds1}
|J(0;z)|^s\leq|\lambda\omega_0+T(0,0)-z|^s+\sum_{n\neq 0}\sum_{m\neq 0}|T(0,m)|^s|G_{\{0\}^c}(m,n;z)|^s|T(0,n)|^s.
\end{equation}

\noindent
2. 
As $|z-T(0,0)| < B |\lambda|$, for some $0<B<+\infty$, independent of $\lambda$, the first term on the right in \eqref{eqJbounds1} satisfies
$$\mathbb{E}(|\lambda\omega_0|^s+|z-T(0,0)|^s)\leq \lambda^s(M(s)+B^s) , 
$$
where $M(s)=\mathbb{E}(|\omega_0|^s)$. For the second term,  
the Aizenman-Molchanov fractional moment \emph{a priori} bound \cite[Corollary 8.4]{AW15}  applied to the depleted resolvent yields
$$
\mathbb{E}(|G_{\{0\}^c}(m,n;z)|^s)\leq \frac{\theta_s}{\lambda^s},
$$
with $\theta_s=\frac{4^s\|\rho\|^s_{\infty}}{1-s}$, see also \cite[Lemma 4.1]{DERM24}. 
The factor involving the hopping term $T$ satisfies  
$$\mathbb{E}(|T(0,m)|^s|G_{\{0\}^c}(m,n)|^s|T(0,n)|^s)\leq \frac{\theta_s}{\lambda^s}|T(0,m)|^s|T(0,n)|^s.$$
Combining these bounds with \eqref{eqJbounds1} yields
\begin{equation}\label{eq:controlJ}
    \mathbb{E}(|J(0;z)|^s)\leq \lambda^s(M(s)+B^s)+\frac{\theta_s\|T\|^{2s}_{s}}{\lambda^s}.
\end{equation}
From\eqref{eq:CS} and \eqref{eq:controlJ}, we infer that
\begin{equation}\label{eq: genlowerapriori1}
\frac{1}{\lambda^s(M(s)+B^s)+\frac{\theta_s\|T\|^{2s}_{s}}{\lambda^s}}\leq \mathbb{E}(|G(0,0;z)|^s).
\end{equation}.

\smallskip

\noindent
3. Similarly, using instead the inverse expression obtained from \eqref{eq:Greeniden1},  
$$\frac{1}{G(0,0;z)}=\lambda\omega_0 - z+\sum_{n\neq 0}\sum_{m\neq 0}T(0,m)G_{\{0\}^c}(m,n)T(0,n), $$
 one achieves the analogous bound
\begin{equation}\label{eq: generalinverse a priori bound}
 \mathbb{E} \left( \frac{1}{|G(0,0;z)|^{s}} \right) \leq \lambda^s(M(s)+B^s)+\frac{\theta_s\|T\|^{2s}_{s}}{\lambda^s}.
\end{equation}
 Combining \eqref{eq: genlowerapriori1} and \eqref{eq: generalinverse a priori bound} we establish \eqref{eq:genaprioribothdir} with $C_{\mathrm{AP}}=\frac{1}{\lambda^s(M(s)+B^s)+\frac{\theta_s\|T\|^{2s}_{s}}{\lambda^s}}$.
\end{proof}
In what follows we will be mostly interested in the large disorder regime, thus we now single out the dependence of $C_{\mathrm{AP}}$ on $\lambda$ in this situation.

\begin{cor}\label{Lem:apriorilower}
Let $s\in (\frac{d}{d+2\beta},1)$ and $B > 0$. There exists a positive threshold $\lambda_{\mathrm{AP}}=\lambda_{\mathrm{AP}}(\rho,s,\|T\|_s,B)$ such that whenever $\lambda>\lambda_{\mathrm{AP}}$ and $z\in \mathbb{C}^+\cap D(T(0,0),\lambda B)$ the following \emph{a priori} bounds hold for all $\Lambda\subset \mathbb{Z}^d$ and $k\in \Lambda$:
\beq\label{eq:dir1}
\frac{A(s,B)}{\lambda^s}\leq \mathbb{E}(|G_{\Lambda}(k,k;z)|^s), 
\eeq
and
\beq\label{eq:inv1}
\mathbb{E}\left(\frac{1}{|G_{\Lambda}(k,k;z)|^s} \right )\leq \frac{\lambda^s}{A(s,B)}  ,
\eeq
where $A(s,B):= (2(\E (|\omega_0|^s)+B^s))^{-1}$. Moreover, we may take \begin{equation}\label{eq:largelambda1}
\lambda_{\mathrm{AP}}^{2s}=\frac{\theta_s\|T\|^{2s}_{s}}{M(s)+B^s},
\end{equation} with
$M(s)=\mathbb{E}(|\omega_0|^s)$
and $\theta_s=\frac{4^s\|\rho\|^s_{\infty}}{1-s}$.
\end{cor}

\begin{proof}
It suffices to recall that in Lemma \ref{Lem:generalapriori}
$$C_{\mathrm{AP}}=\frac{1}{\lambda^s(M(s)+B^s)+\frac{\theta_s\|T\|^{2s}_{s}}{\lambda^s}}.$$
Defining $\lambda_{\mathrm{AP}}$ by
\begin{equation}
\lambda_{\mathrm{AP}}^{2s}=\frac{\theta_s\|T\|^{2s}_{s}}{M(s)+B^s},
\end{equation}
we see that for $\lambda>\lambda_{\mathrm{AP}}$ 
$$ \frac{A(s,B)}{\lambda^s}\leq C_{\mathrm{AP}}$$
where $A(s,B):= ({2(M(s)+B^s)})^{-1}$. Recalling Lemma \ref{Lem:generalapriori} we obtain the bounds
 \eqref{eq:dir1} and \eqref{eq:inv1}.

\end{proof}

The estimates in Corollary \ref{Lem:apriorilower} imply a useful decoupling result.

\begin{lem}\label{Lem:decouplinglowerbounds}
Fix $\beta>\frac{d}{2}$, $B>0$, $s\in (\frac{2d}{d+2\beta},1)$ and let $A(s,B)$ be as in Corollary \ref{Lem:apriorilower}. There exists a positive threshold $\lambda_{\mathrm{AP}}=\lambda_{\mathrm{AP}}(\rho,s,\|T\|_s,B)$ such that whenever $\lambda>\lambda_{\mathrm{AP}}$ and $z\in \mathbb{C}^+\cap D(T(0,0),\lambda B)$, the following 
\textit{a priori} bound holds for all $\Lambda_1,\Lambda_2\subset \mathbb{Z}^d$ and $k \in \Lambda_1$ , $j \in \Lambda_2$:
$$\frac{A^2(s/2,B)A(s,B)}{\lambda^{2s}}\leq \mathbb{E}(|G_{\Lambda_1}(k,k;z)|^s|G_{\Lambda_2}(j,j;z)|^s).$$
\end{lem}

\begin{proof}
By the Cauchy-Schwarz inequality, 
\begin{align*}
 \mathbb{E}(|G_{\Lambda_1}(k,k;z)|^{s/2}) &\leq \mathbb{E}^{1/2}(|G_{\Lambda_1}(k,k;z)|^s|G_{\Lambda_2}(j,j;z)|^s)\\
 &\times\mathbb{E}^{1/2}(\frac{1}{|G_{\Lambda_2}(j,j;z)|^s}).\\
\end{align*}
We conclude from Corollary \ref{Lem:apriorilower} that
\begin{equation}
    \frac{A^2(s/2,B)A(s,B)}{\lambda^{2s}}\leq \mathbb{E}(|G_{\Lambda_1}(k,k;z)|^s|G_{\Lambda_2}(j,j;z)|^s),
\end{equation}
establishing the result.
\end{proof}

\subsection{Main results on Green's function lower bounds}\label{subsec:lowerBd_main1}

After the preliminary results on the fractional moments of diagonal elements of the Green's function in section \ref{subsec:apriori_lb1}, we prove our main theorem on a lower bound estimate for the expectation of the fractional moments of the Green's function. We will use the decoupling estimate in Lemma \ref{Lem:decouplinglowerbounds} and the self-avoiding walk (SAW), or Feenberg loop-erased expansion, of the Green's function $G^\Lambda (z; 0, n)$, for $0, n \in \Lambda$, see \cite[section 6.2]{AW15}. Let $\gamma : K \subset \N_0 \to \Lambda$ be a self-avoiding path in $\Lambda$. That is, the map $\gamma: K \subset \N_0 \to \Lambda \subset \Z^d$ is injective. In an abuse of notation, we call a self-avoiding path, a SAW. We consider the SAW $\gamma(K) \subset \Lambda$ with $\gamma(0) = 0$ and $\gamma(|\gamma|) = n$, where $|\gamma|$, the path length, denotes the number of steps in the path. The symbol $\sum^{SAW}_{\gamma: 0\mapsto n}$ denotes summation over all such self-avoiding paths $\gamma$ in $\Lambda$ which start at $0$ and end at $n$.  The maximal length of a SAW in $\Lambda$ is $|\Lambda|$. The SAW expansion (see, for example, \cite[Theorem 6.2]{AW15}) has the form
\beq\label{eq:saw1}
G_{\Lambda}(z; 0,n) = \sum^{SAW}_{\gamma: 0 \mapsto n}(-1)^{|\gamma|}\prod^{|\gamma|}_{j=1}T(\gamma(j-1),\gamma(j)) \prod^{|\gamma|}_{k=0} 
{G}_{\Lambda \backslash \gamma([0, k-1])} ( \gamma(k),\gamma(k);z) ,
\eeq
where 
\beq
{G}_{\Lambda \backslash \gamma([0, k-1])}  ( \gamma(k),\gamma(k);z)  = \langle \delta_{\gamma(k)}, (H_{\Lambda \backslash \gamma([0, k-1])}-z)^{-1}\delta_{\gamma(k)}\rangle, 
\eeq
with the convention that for $k=0$, we have $H_{\Lambda \backslash \gamma([0, -1])} = H^\Lambda$. 
Note that the SAW expansion expresses off-diagonal matrix elements in terms of diagonal matrix elements, which can be estimated as in Lemma \ref{Lem:decouplinglowerbounds}. We also observe that, since $T$ has nonzero matrix elements $T(0,k)$ for all $k \in \Z^d\setminus\{0\}$, the smallest SAW from $0$ to $n$ has one step.

\medskip

\begin{thm}\label{thm:lb_green1}
Let $\Lambda \subset \Z^d$, $B>0$ and $0, n \in \Lambda$ with $n\neq 0$. Fix $\beta>\frac{d}{2}$, and $s\in (\frac{2d}{d+2\beta},1)$. There exists $\lambda_1=\lambda_1(\beta,d,\rho,s,\|T\|_s,B)>0$ such that for $\lambda>\lambda_1$ and $z\in \mathbb{C}^+\cap D(T(0,0),\lambda B)$  there exists $m(\beta,d,B,s)>0$ such that
\begin{equation}\label{eq:lb_green1}
 \left(\frac{m(\beta,d,B,s)}{\lambda^{2s}}\right)\frac{1}{\|n\|^{s(d+2\beta)}}\leq \mathbb{E}(|G_{\Lambda}( 0,n;z)|^s).
\end{equation}
\end{thm}

%

\begin{proof}
1. For any path $\gamma$, let 
$$
\mathcal{T}(\gamma) := (-1)^{|\gamma|}\prod^{|\gamma|}_{j=1}T(\gamma(j-1),\gamma(j)). 
$$
There is exactly one path of length one in the SAW sum \eqref{eq:saw1} given by $\gamma =   \{ \gamma (0) = 0, \gamma (1) = n \}$. 
Separating this term from the remainder of the sum we obtain 
\bea\label{eq:saw2}
G_{\Lambda}(0,n;z) & = & - T(0,n) G_{\Lambda}(0,0;z) {G}_{\Lambda \backslash \{0\}}(n,n;z)  \nonumber \\
 &+ &  \sum^{SAW}_{\substack{\gamma: 0 \mapsto n\\|\gamma| \geq 2}} \mathcal{T}(\gamma) 
 \prod^{|\gamma|}_{k=0} 
{G}_{\Lambda \backslash \gamma([0, k-1])} ( \gamma(k),\gamma(k) ;z),
\eea 
Rearranging, and taking the $s$-power, we reach 
\bea\label{eq:expansioninequality}
\lefteqn{|T(0,n)|^s|G_{\Lambda}(0,0;z)|^s |{G}_{\Lambda \backslash \{0\}}(n,n;z)|^s } \nonumber \\
 & \leq  & |G_{\Lambda}(0,n;z)|^s    \nonumber \\
 &  & + \sum^{SAW}_{\substack{\gamma: 0\mapsto n\\|\gamma|\geq 2}} 
 | \mathcal{T}(\gamma) |^s \prod^{|\gamma|}_{k=0}|{G}_{\Lambda \backslash \gamma([0, k-1])}(\gamma(k),\gamma(k);z)|^s. \nonumber \\
  & 
\eea

\noindent
2. From definition of $T$, Hypothesis \ref{hypothesis:lr1} and Lemma \ref{Lem:decouplinglowerbounds} , we achieve a lower bound for the left side of \eqref{eq:expansioninequality}, 
\begin{equation}\label{mainterm}
\frac{c_{\beta,d}}{\|n\|^{s(d+2\beta)}} \left( \frac{A^2(s/2,B)A(s,B)}{\lambda^{2s}}\right) \leq \mathbb{E}(|T(0,n)|^s|G(0,0;z)|^s  |  {G}_{\Lambda \backslash \{0\} }(n,n;z)|^s).
\end{equation}
We next estimate from above the second term on the right in \eqref{eq:expansioninequality}. The local \emph{a priori} fractional moment bound \cite[Corollary 8.4]{AW15} gives
$$
\mathbb{E}_{\omega(\gamma(k))} 
(|  {G}_{\Lambda \backslash \gamma([0, k-1])}(\gamma(k),\gamma(k);z)|^s)\leq \frac{\theta_s}{\lambda^s}, 
$$
which allows us to establish the following upper bound
$$
\sum^{SAW}_{\substack{\gamma: 0\mapsto n\\|\gamma|\geq 2}}
|\mathcal{T}(\gamma)|^s \mathbb{E} \left( \prod_{k=0}^{|\gamma|}
|  {G}_{\Lambda \backslash \gamma([0, k-1])}(\gamma(k),\gamma(k);z)|^s \right)
\leq 
\sum^{\mathrm{|\Lambda|}}_{l=2} \sum^{SAW}_{\substack{\gamma: 0\mapsto n\\|\gamma|=l}} | \mathcal{T}(\gamma) |^s \left( \frac{\theta_s}{\lambda^s}\right)^{l+1}.$$

\noindent
3. We estimate the sum over paths over length $l$ starting at $0$ and ending at $n$ by noting that each path of this form is determined by choosing $l-1$ distinct points
$n_1,\ldots,n_{l-1}$ in $\Lambda$. Setting $n_0=0,n_l=n$ and using that
$\|n\|\leq  \sum^l_{k=1} \| n_{k} -n_{k-1}\|$, we find that  for at least one value of $k^*\in \{1,\ldots,l-1\}$,
$$
\| n_{k^*+1}-n_{k^*} \| \geq \frac{\|n\|}{l} .
$$
For such $k^*$, this yields
$$
T(n_{k^*} ,n_{k^*+1}) \leq  \frac{C_{d,\beta}l^{d+2\beta}}{\|n\|^{d+2\beta}}.
$$
Recalling the definition of $\| T \|_s^s$ in \eqref{eq:T_defn1}, one has
$$
\sum^{SAW}_{\substack{\gamma: 0\mapsto n\\|\gamma|=l}}
\left( \prod^{l}_{\substack{j=1 \\ j \neq k^* +1}} |T(\gamma(j-1),\gamma(j))|^s \right) \leq (\|T\|^s_{s})^{l-1} . 
$$ 
If $\lambda$ is sufficiently large so that $\| T \|_s^s \theta_s \lambda^{-s} < 1$, we find from the above that
\bea\label{eq:remainder}
\sum^{\mathrm{|\Lambda|}}_{l=2} \sum^{SAW}_{\substack{\gamma: 0\mapsto n\\|\gamma|=l}} | \mathcal{T}(\gamma)|^s 
\left(\frac{\theta_s}{\lambda^s}\right)^{l+1} & \leq  & 
  \sum^{\mathrm{|\Lambda|}}_{l=2}(\|T\|^s_{s})^{l-1}\frac{C^s_{d,\beta}l^{s(d+2\beta)}}{\|n\|^{s(d+2\beta)}}\left(\frac{\theta_s}{\lambda^s}\right)^{l+1}
  \nonumber \\
  &\leq & \frac{D(s,d,\beta,\|T\|^s_{s})}{\lambda^{3s} \|n\|^{s(d+2\beta)}}.
\eea
for all $\lambda>\lambda(s,\rho,\|T\|^s_{s})$, for some $\lambda(s,\rho,\|T\|^s_{s})>0$. 

\noindent
4. Combining equations \eqref{eq:expansioninequality},\eqref{mainterm} and \eqref{eq:remainder}, we obtain
\begin{equation}
    \frac{c_{\beta,d}}{|n|^{s(d+2\beta)}}
     \left( \frac{A^2(s/2,B)A(s,B)}{(\lambda^{s})^2} \right)  \leq \mathbb{E}(|G_{\Lambda}(0,n;z)|^s)+\frac{D(s,d,\beta,\|T\|^s_{s})}{(\lambda^s)^3\|n\|^{s(d+2\beta)}}
\end{equation}
Taking $\lambda$ large enough, also depending on $A(s,B)$ and $D(s,d,\beta,\|T\|^s_{s})$, such that
\be  \label{eq:secondlargedisorderlowerb}
c_{\beta,d}A^2(s/2,B)A(s,B) -  \frac{D(s,d,\beta,\|T\|^s_{s})}{\lambda^s}>\frac{1}{2}c_{\beta,d}A^2(s/2,B)A(s,B), \ee 
and defining
\be  \label{eq:mdefinition}
m(\beta,d,B,s):=\frac{1}{2}c_{\beta,d}A^2(s/2,B)A(s,B), \ee 
we arrive at a positive lower bound for the fractional moment of the Green's function:  
\begin{equation}
0 < \left(\frac{m(\beta,d,B,s)}{\lambda^{2s}}\right)\frac{1}{\|n\|^{s(d+2\beta)}}\leq \mathbb{E}(|G_{\Lambda}(0,n;z)|^s)
\end{equation}
whenever
$n\neq 0,\,\,\,s\in (\frac{2d}{d+2\beta},1)$, proving \eqref{eq:lb_green1}.
\end{proof}


\begin{cor} \label{cor:lb_green2} 
Under the conditions of Theorem \ref{thm:lb_green1}, the Green's function second moments satisfy 
\begin{equation}\label{eq:lb_green_moment1} 
0 < \left(\frac{m^{2/s}(\beta,d,B,s)}{\lambda^{4}}\right)\frac{1}{|n|^{2(d+2\beta)}}\leq \mathbb{E}(|G_{\Lambda}(0,n;z)|^2) .
\end{equation}
\end{cor}

\begin{proof}
By H\"older's inequality
\begin{equation}
    \mathbb{E}(|G_{\Lambda}(0,n;z)|^s)\leq \mathbb{E}(|G_{\Lambda}(0,n;z)|^2)^{s/2}.
\end{equation}
Thus, from the bound \eqref{eq:lb_green1}, we get the desired result. 
\end{proof}

We note a first application of the lower bound \eqref{eq:lb_green_moment1} to the question of localization. If $\beta>d/2$, the lower bound implies that 
$$
\sum_{n\in \Lambda}|n|^q\mathbb{E}(|G_{\Lambda}(0,n;z)|^2) \geq
\left(\frac{m^{2/s}(\beta,d,B,s)}{\lambda^{4}}\right) \left( \sum_{n\in \Lambda}\frac{1}{\|n\|^{2d+4\beta-q}} \right) .
$$
The above sum diverges as $|\Lambda|\to \infty$ whenever
$d+4\beta-q\leq0$,  i.e.\ for 
$\frac{d}{2}<\beta\leq \frac{q-d}{4}$. Thus, there is a nontrivial interval of $\beta's$ where this holds when $q>3d$, and this is the main reason behind the delocalization statement in Theorem \ref{thm:no_dl1}.

\subsection{Infinite volume Green's function bounds}

We now extend the finite-volume results to the infinite-volume case using the second resolvent formula and the properties of the hopping term $T$. 

\begin{lem}\label{Lem:infinitevolcomparision}
Let $\Lambda=[-L,L]^d\cap \mathbb{Z}^d$, $\beta>\frac{d}{2}$, $z\in \mathbb{C}^+$. Moreover, assume that $\lambda>\lambda_0$ with $\lambda_0=\lambda_0(\beta,s,d)$ as in Theorem \ref{thm:decay-long-range-GF}. Then, for each $s\in(\frac{2d}{d+2\beta},1)$ we have the convergence
    $$\lim_{L\to \infty} \mathbb{E}(|G_{\Lambda}(0,n;z)-G(0,n;z)|^{s})=0.$$
\end{lem}
\begin{proof}
The geometric resolvent equation gives the bound 
\begin{align*}
   |G_{\Lambda}(0,n;z)-G(0,n;z)|&\leq \sum_{(l,m)\in \Lambda\times \Lambda^c}|G_{\Lambda}(0,l;z)T(l,m)G(m,n;z) . 
\end{align*}
We take the expectation of the $s$-power and use the bounds on the kernel of $T$ in Hypothesis \eqref{hypothesis:lr1} along with the Green's function decay of Theorem \ref{thm:decay-long-range-GF} to obtain 
\bea \label{Eq:infinitevolcomparision}
\lefteqn{\mathbb{E}(|G_{\Lambda}(0,n;z)-G(0,n;z)|^{s}) } \nonumber \\
 & \leq & \frac{1}{ (\Im z )^{s}}  \sum_{(l,m)\in \Lambda\times \Lambda^c}  |T(l,m)|^{s}   \mathbb{E} ( |G_{\Lambda} (0,l;z)|^{s})^{\frac{1}{2}}   
     \mathbb{E} (|G(m,n;z)|^{s})^{\frac{1}{2}} \nonumber \\
 & \leq &\frac{C_{\beta,d,s,\rho,\lambda}}{(\Im z)^{s} }   
    \sum_{(l,m)\in \Lambda \times \Lambda^c}f(m,l)\frac{1}{\|m-n\|^{\frac{s(d+2\beta)}{2}}},   \nonumber 
    \\
\eea
where $$f(m,l):=\left\{ \begin{array}{cc} 
 \frac{1}{\|l-m\|^{s(d+2\beta)}}
    \frac{1}{\|l\|^{\frac{s(d+2\beta)}{2}}}, & l \neq 0; \\
 \frac{1}{\|m\|^{s(d+2\beta)}}
    , \ l = 0.
  \end{array} \right.$$
  Since $m\in \Lambda^c$,
  $f(m,l)\leq (\frac{2}{L})^{{\frac{s(d+2\beta)}{2}}}$ for any $l\in \Lambda$ and thus \eqref{Eq:infinitevolcomparision} yields
    \begin{equation*}
         \mathbb{E}(|G_{\Lambda}(0,n;z)-G(0,n;z)|^{s})
  \leq 2^{{\frac{s(d+2\beta)}{2}}}\frac{L^d}{L^{\frac{s(d+2\beta)}{2}}}\sum_{m\in \Lambda^c}
   \frac{1}{\|m-n\|^{s(d+2\beta)}}. 
\end{equation*}
   
By the assumption on $s$ this vanishes as $L \to \infty$, finishing the proof.
\end{proof}

As a consequence of Theorem \ref{thm:lb_green1} and Lemma \ref{Lem:infinitevolcomparision}, we obtain a lower bound for the expectation of the $s$-power of the Green's function. 

\begin{cor}\label{corMM}
Let $\beta>\frac{d}{2}$, $B>0$, and $s\in (\frac{2d}{d+2\beta},1)$ be given.
There exists $\lambda_1=\lambda_1(\beta,d,\rho,s,\|T\|_s,B)>0$ and $m(\beta,d,B,s)>0$ such that for $\lambda>\lambda_1$
the infinite volume Green's function for the lrAM satisfies the lower bound:   \begin{equation}
\left(\frac{m(\beta,d,B,s)}{2\lambda^{2s}}\right)\frac{1}{\|n\|^{s(d+2\beta)}}\leq \mathbb{E}(|G(0,n;z)|^s)
\end{equation}
whenever
$n\neq 0$, and $z\in \mathbb{C}^{+}\cap D(T(0,0),\lambda B)$.
\end{cor}
As before, this immediately implies the following corollary.

\begin{cor}\label{secondmomentcor}
Let $\beta>\frac{d}{2}$, $B>0$, and $s\in (\frac{2d}{d+2\beta},1)$ be given.
There exists $\lambda_1=\lambda_1(\beta,d,\rho,s,\|T\|_s,B)>0$ and $m(\beta,d,B,s)>0$ such that for $\lambda>\lambda_1$ the infinite volume Green's function for the lrAM satisfies the lower bound:   \begin{equation}\label{eq:lb_green_2}
\left(\frac{m^{\frac{2}{s}}(\beta,d,B,s)}{2^{\frac{2}{s}}\lambda^{4}}\right)\frac{1}{\|n\|^{2(d+2\beta)}}\leq \mathbb{E}(|G(0,n;z)|^2)
\end{equation}
whenever
$n\neq 0$, and $z\in \mathbb{C}^{+}\cap D(T(0,0),\lambda B)$.
\end{cor}


\section{Dynamical delocalization and the absence of dynamical localization for fAM and lrAM}
\label{sec:deloc1}

We recall that a convenient measure of transport is given by the time and disorder average of moments of the position operator with respect to a solution of the \Schr equation
with initial condition $\varphi_0$. This solution is given by $\varphi(t) := e^{-it H_\beta} \varphi_0$. For simplicity, we take $\varphi_0 = \delta_0$ and denote  this moment by $M_T^q$: 
\beq\label{defn:moment1}
M^{q}_{T }=\frac{2}{T}\int^{\infty}_{0}e^{\frac{-2t}{T}}\mathbb{E}\langle \delta_0,{e^{itH_{\beta}}|X|^q e^{-itH_{\beta}}}\delta_0\rangle\,dt.
\end{equation} 

In the spirit of Definition \ref{def:dyn-loc}, if any of these moments are unbounded as $T \to \infty$ we say that the system exhibits \emph{dynamical delocalization}. This implies that there is nontrivial transport in the system. 

We recall that spectral localization with polynomially decaying eigenfunctions holds for the fAM and lrAM at strong disorder, as shown in Section \ref{sec:long_range1}. The most important consequence of  Theorem \ref{thm:lb_green1} on the lower bounds for the Green's function 
is the absence of dynamical localization for these models, which is in stark contrast with the behavior exhibited by the usual Anderson model. 


\begin{thm}\label{thm:no_dl1}
Let $M_T^q$ be the $q$-th moment of the position operator for the lrAM satisfying Hypotheses \ref{hypothesis:pot1} and \ref{hypothesis:lr1}, with long-range hopping exponent $\beta > \frac{d}{2}$, and $s \in ( \frac{2d}{d + 2 \beta}, 1)$. For any dimension $d \geq 1$, there exists $\lambda_1=\lambda_1(\beta, d,\rho,s,\|T\|_s)>0$ such that for $\lambda>\lambda_1$:
\begin{enumerate}  
\item The small moments are finite: $M^{q}_{T}<\infty$, if $\beta>0$ and $0 < q<2\beta$; 

\medskip

\item The  higher moments are infinite: $M^{q}_{T}=+\infty$, if $\beta>d/2$ and $q\geq d+4\beta$.
\end{enumerate}
Consequently, under the above conditions, the lrAM does not exhibit dynamical localization but exhibits dynamical delocalization. 
In the particular, this applies to the fractional Anderson Model in $d=1$ with exponent $\frac{1}{2}<\alpha<1$.
\end{thm}


\begin{proof}
The bound in case 1 
follows from \cite[Theorem 2]{DERM24} and Theorem \ref{thm:spectral-loc-lram}.  As for case 2, we express the time-average in \eqref{defn:moment1} in terms of the Fourier transform with help of Plancherel's identity
\begin{equation}\label{eq:Plancherelm}
\int^{\infty}_{0}e^{\frac{-2t}{T}}
\mathbb{E}(|\langle \delta_n,e^{itH_{\beta}}\delta_0\rangle|^2)\,dt=\frac{1}{2\pi}\int_{\mathbb{R}}\mathbb{E}(|G(n,0;E+\frac{i}{T})|^2)\,dE
\end{equation}

and then apply the result of Corollary \ref{secondmomentcor}. This gives 
 \begin{align}
 \nonumber M^{q}_{T}&=\frac{2}{ T}\sum_{n\in \mathbb{Z}^d}\|n\|^q\int^{\infty}_{0}e^{\frac{-2t}{T}}
\mathbb{E}(|\langle \delta_n,e^{itH_{\beta}}\delta_0\rangle|^2)\,dt\\
&=\frac{1}{\pi T}\sum_{n\in \mathbb{Z}^d}\|n\|^q\int_{\mathbb{R}}\mathbb{E}(|G(n,0;E+\frac{i}{T})|^2)\,dE\\
\nonumber&\geq \frac{|I|}{\pi T}\left(\frac{m^{2/s}(\beta,d,B,s)}{2^{\frac{2}{s}}\lambda^{4}}\right)\sum_{n\in \mathbb{Z}^d}\frac{\|n\|^q}{|n|^{2(d+2\beta)}}
\end{align}
where we have denoted by $I$ an interval with length comparable to $B \lambda$ which contains $\sigma(H_{\beta})$. Now it suffices to observe that
$$
 \sum_{n\in \mathbb{Z}^d \backslash \{0\}} \|n\|^{q- (2d+4\beta)} =+\infty,
$$
 if $q\geq d+4\beta$.
\end{proof}

\subsection{Lower bounds for eigenfunction correlators}
\label{subsec:lb_ec}

We conclude with some lower bound estimates on the eigenfunction correlators (see, for example, \cite[Chapter 7]{AW15}). We recall the definition of $Q(k,n,I)$, for $n,k \in \Z^d$, 
and any interval $I\subset \mathbb{R}$:
\beq\label{defn:correl1}
Q(k,n;I) :=\sup_{\substack{|f|\leq 1 \\ \supp f \subset I}} |\langle\delta_n, P_I(H_\beta) f(H_{\omega})\delta_k\rangle|,
\eeq
with the above supremum being taken over Borel measurable functions. If $I \cap \Sigma_\beta = \emptyset$, we have that $Q(k,n;I ) = 0$.
If $f(x) = {1}_I (x)$, and $\varphi_j(x)$ are the normalized eigenfunctions of $H_\beta$
with corresponding eigenvalues $E_j$, then 
\beq 
   Q(k,n;I) = \sum_{j; E_j \in I} | \varphi_j(k) \varphi_j(n)| .
   \eeq
In particular, bounds on the expectation of the correlator $Q(0,n;I)$  suggest an averaged bound on the eigenfunctions in the sum. We present upper and lower bounds on the global eigenfunction correlators of $H_{\beta}$ which hold in expectation. We are only able to prove these bounds for any interval $I \subset \R$ for which $\Sigma_\beta \subset I$. For this reason, we refer to these as \textit{global eigenfunction correlators} and we simply write $Q(0,n):= Q(0,n;I)$, for such $I$.

\begin{cor}\label{cor:ef-corr}
The global eigenfunction correlator $Q(0,n)$ is bounded above and below as follows:
    \begin{enumerate}
      \item \label{upperboundcorr}
      Assume that $\beta>0$ and $\lambda_0>0$ is as in Theorem \ref{thm:decay-long-range-GF} . Then for each $s\in (\frac{d}{d+2\beta},1)$, $\lambda>\lambda_0$ and $n\neq 0$, the upper bound
      $$\mathbb{E}\left(Q(0,n)\right)\leq \frac{C_{\lambda,d,\beta,s}}{\|n\|^{s(d+2\beta)}} , 
      $$
      holds with $C_{\lambda,d,\beta,s}>0$.
    \item \label{lowerboundcorr}
    Assume that $\beta>\frac{d}{2}$. Let $\lambda_1>0$ and $m(\beta,d,B,s)>0$ be as in Theorem \ref{thm:lb_green1}. Let $\varepsilon>0$, fix $\Sigma_\beta \subset I\subset \mathbb{R}$, a bounded interval, and let $B$ be chosen so that
    $\max\{\varepsilon,|E-T(0,0)|\}<B\lambda$
     for all $E\in|I|$. Then, for each $s\in (\frac{2d}{d+2\beta},1)$ and $n\neq 0$, the lower bound holds
    $$\mathbb{E}^{1/2}\left(Q^2(0,n)\right)\geq \left(\frac{\sqrt{\varepsilon|I|}m^{1/s}(\beta,d,B,s)}{2^{1/s}\sqrt{\pi}\lambda^{2}}\right)\frac{1}{\|n\|^{(d+2\beta)}} .
    $$
  \end{enumerate}
\end{cor}

\begin{proof}
The first estimate follows from \cite[Theorem 2]{DERM24}. Regarding the lower bound, observe that by definition of the eigenfunction correlators \eqref{defn:correl1}
$$
Q^2(0,n) \geq \varepsilon \int_0^\infty |\langle \delta_0, P_I(H_\beta) e^{-(\varepsilon - iH_\beta) t} \delta_0\rangle|^2) ~dt.
$$
Thus, it follows from Plancherel's Theorem, as in \eqref{eq:Plancherelm}, that 
$$
Q^2(0,n) \geq \frac{\varepsilon}{2\pi}\int_{\mathbb{R}} |G(n,0;E+i\varepsilon)|^2 ~dE ,
$$
since $P_I(H_\beta) (H_\beta - z)^{-1} = (H_\beta - z)^{-1}$. 
Restricting the above integral to the interval $I$, and taking the expectation, we obtain
$$
\mathbb{E}\left(Q^2(0,n;I)\right) \geq \frac{\varepsilon}{2\pi}\int_{I}\mathbb{E}(|G_{}(n,0;E+i\varepsilon)|^2)\,dE, $$
so that
$$\mathbb{E}(Q^2(0,n;I))\geq \frac{\varepsilon}{\pi} |I|\inf_{E\in I} \mathbb{E}(|G(n,0;E+i\varepsilon)|^2) . 
$$
Since the expectation of the square of the Green's function is bounded below as in \eqref{eq:lb_green_2}, we obtain 
$$
\mathbb{E}\left(Q^2(0,n;I)\right)\geq \frac{\varepsilon}{2\pi} |I|\left(\frac{m^{2/s}(\beta,d,B,s)}{2^{2/s}\lambda^{4}}\right)\frac{1}{\|n\|^{2(d+2\beta)}}.
$$
This concludes the proof of the lower bound \eqref{lowerboundcorr}.
\end{proof}

\begin{rem}
In the above result $s$ may be taken arbitrarily close to one. The upper and lower bounds in parts \ref{upperboundcorr} and \ref{lowerboundcorr} of Corollary \ref{cor:ef-corr} 
suggest that the global eigenfunction correlator $Q(0,n)\sim \frac{1}{\|n\|^{(d+2\beta)}}$ should hold in an appropriate sense.
Indeed, in the presence of suitable large deviation estimates for the eigenfunction correlator, the above results should yield the conjectured behavior of $\E ( Q(0,n) )$.
\end{rem}

\section*{Acknowledgements}
PDH is partially supported by Simons Foundation Collaboration Grant for Mathematicians No.\ 843327.
CRM is grateful to P. M\"uller, F. Nakano and M. Disertori for enlightening discussions, and to PUC Rio, Brazil, for their hospitality. CRM acknowledges the support of the Fondation de Sciences de la Mod\'elisation and CY Initiative.
RM thanks CNPq for partial support under
grants 308527/2026-7, 402952/2023-5 and 402249/2024-0. 



\end{document}